\documentclass[pdflatex,sn-mathphys,Numbered]{sn-jnl}

\usepackage{graphicx}%
\usepackage{multirow}%
\usepackage{amsmath,amssymb,amsfonts}%
\usepackage{amsthm}%
\usepackage{mathtools}%
\usepackage{bm}%
\usepackage{mathrsfs}%
\usepackage[title]{appendix}%
\usepackage{xcolor}%
\usepackage{textcomp}%
\usepackage{manyfoot}%
\usepackage{booktabs}%
\usepackage{enumitem}%
\usepackage{tikz}%
\usetikzlibrary{arrows.meta,positioning,fit,backgrounds,calc}%
\usepackage[nameinlink,capitalise]{cleveref}%

\AtBeginDocument{}
\makeatletter
\g@addto@macro\UrlBreaks{\do\_\do\/\do\-\do\.\do\=\do\?}
\makeatother
\newcommand{\N}{\mathcal{N}}
\newcommand{\Part}{\mathfrak{P}}
\newcommand{\E}{\mathbb{E}}
\newcommand{\R}{\mathbb{R}}
\newcommand{\btheta}{\bm{\theta}}
\newcommand{\calD}{\mathcal{D}}
\newcommand{\calV}{\mathcal{V}}
\newcommand{\PiNS}{\Pi^{\mathrm{NS}}}

\newtheorem{theorem}{Theorem}[section]
\newtheorem{proposition}[theorem]{Proposition}
\newtheorem{corollary}[theorem]{Corollary}

\theoremstyle{definition}
\newtheorem{definition}[theorem]{Definition}
\newtheorem{assumption}[theorem]{Assumption}
\theoremstyle{remark}
\newtheorem{remark}[theorem]{Remark}

\begin{document}

\title[Stable, Budget-Feasible Coalitions for Clustered FL]{Stable and Budget-Feasible Coalition Formation for Clustered Federated Learning: A Hedonic Potential-Game Approach}

\author[1]{\fnm{Cengis} \sur{Hasan}}\email{cengis.hasan@cognifinity.lu}

\affil[1]{\orgname{Cognifinity}, \orgaddress{\country{Luxembourg}}}

\abstract{Federated learning systems can benefit from organizing heterogeneous participants into coalitions that train coalition-specific models. Such clustering is sustainable only if participants prefer to remain in their assigned coalition and the associated transfers are affordable. We develop a transferable-surplus model that separates a coalition's learning benefit, system cost, participant cost, and monetary transfers. An allocation rule converts coalition surplus into hedonic preferences, while weak budget feasibility guarantees nonnegative retained coordinator surplus. For symmetric pairwise surplus allocations, we prove that the induced coalition-formation game is an exact potential game. Consequently, a Nash-stable partition exists, and every strict better-response process terminates at such a partition. When destination members may reject entrants, accepted better responses instead terminate at an individually stable partition. We characterize feasibility of bounded pair incentives and give a polynomial verification result when retained budget slack is submodular. We then decompose welfare into participant potential and coordinator-retained slack and derive additive and relative-slack multiplicative efficiency guarantees. Exact budget balance yields a welfare-optimal Nash-stable partition on the exactly pairwise-representable surplus class, whereas budget feasibility alone permits unbounded welfare loss. Finally, global potential maximization is weighted maximum-agreement correlation clustering up to a constant. Approximation followed by strict better-response stabilization preserves the agreement guarantee and yields an explicit end-to-end welfare bound depending on retained slack and negative-edge mass. Explicit constructions show that the relative-slack guarantee is asymptotically tight and that the negative-edge correction can be attained. In a preregistered five-seed CIFAR-10 study, the proposed mechanism reaches the certified estimated-table welfare optimum on every primary instance, while equal-surplus sharing has no Nash-stable outcome on three instances; pairwise validation gain also gives substantially more reliable pair signs than gradient alignment. These results connect learning value, monetary transfers, stability, and economic efficiency without conflating local equilibrium, global optimality, and computational tractability.}



\maketitle

\section{Introduction}
Federated learning allows multiple participants to train models without pooling their raw data at a central location \cite{McMahanEtAl2017FedAvg,LiEtAl2020Survey,KairouzEtAl2021OpenProblems}. Its practical performance nevertheless depends on which participants collaborate. Differences in data distributions, sample sizes, communication reliability, and participation costs can make a single grand coalition undesirable. Clustered or personalized federated learning addresses the statistical side of this problem by allowing groups of compatible participants to train different models. The economic side remains equally important: participants must be compensated, the coordinating entity must be able to finance the compensation, and no participant should prefer an available unilateral move to another coalition.

This paper studies the joint economic and strategic foundations of coalition formation in clustered federated learning. A coalition trains a coalition-specific model and produces an expected learning benefit. After system and participation costs are deducted, the remaining transferable surplus can be divided among the participants and the coordinating entity. The participant shares induce preferences over coalitions. We ask whether a designer can select an affordable allocation rule that admits a stable coalition partition and whether such a partition can be reached through decentralized decisions.

The problem is naturally modeled as a hedonic coalition-formation game, in which each participant's preference depends on the membership of its own coalition \cite{BogomolnaiaJackson2002,Hajdukova2006,AzizBrandl2012}. We use Nash stability as a strong baseline: no participant can improve its utility by leaving its coalition and joining another existing coalition, including the empty coalition that represents operating alone. Because Nash stability gives the destination no veto, we also analyze individual stability, under which every destination member must weakly accept an entrant.

The paper makes seven contributions.
\begin{enumerate}[leftmargin=*]
    \item It introduces a coalition-specific federated learning model that handles unsuccessful communication rounds without an undefined aggregation and does not assume that merging coalitions always lowers loss.
    \item It separates learning benefit, coordinator cost, participant cost, transfers, participant utility, and retained coordinator surplus. This separation yields explicit budget-feasibility and individual-rationality conditions.
    \item It characterizes bounded symmetric pairwise allocations, proves an exact potential representation, and establishes existence and finite convergence results for both Nash stability and destination-consent-based individual stability.
    \item It links stability to efficiency through coordinator-retained budget slack, giving additive and relative-slack multiplicative guarantees. The relative-slack bound is asymptotically tight. Exact balance gives a welfare-optimal stable partition on the exactly pairwise-representable class, while budget feasibility alone gives no finite welfare guarantee.
    \item It shows that exponentially many budget constraints can be verified in polynomial oracle time when the retained-slack set function is submodular, including a sufficient case with submodular surplus and nonnegative pair rewards.
    \item It identifies global potential maximization with weighted maximum-agreement correlation clustering and derives an end-to-end welfare guarantee for approximation followed by stability post-processing, with an instance attaining its negative-edge correction.
    \item It evaluates the complete design on heterogeneous CIFAR-10 instances. At the primary calibration, the decentralized mechanism reaches the certified estimated-table welfare optimum for all five seeds, whereas equal-surplus sharing has an empty Nash-stable set for three seeds. The experiments also quantify stability costs, retained slack, negative-edge mass, estimator sign errors, and uncertainty in the estimated coalition values.
\end{enumerate}

The present formulation corrects three limitations of the preliminary conference version \cite{Hasan2021HedonicFL}.\footnote{A preliminary version of this work appeared at the Workshop on Optimization and Learning in Multi-Agent Systems (OptLearnMAS) at AAMAS~2021.} First, coalition membership now affects the trained model because each coalition trains its own model. Second, superadditivity of total coalition surplus is not treated as sufficient for individual preferences over larger coalitions. Third, a Nash equilibrium is correctly interpreted as a unilateral local optimum of the potential rather than necessarily a global optimum.

\section{Related Work}
Research on incentives for federated learning has used auctions, contracts, Stackelberg games, reinforcement learning, and privacy-aware mechanisms \cite{LeEtAl2021Auction,DingEtAl2020PrivateInfo,SarikayaErcetin2020,ZhanZhang2020DRL,Kim2020IncentiveDP}. An economic and game-theoretic survey is given in \cite{TuEtAl2022Incentive}. Recent cross-silo work jointly considers data diversity, truthful participant selection, and a limited reward budget \cite{WuEtAl2025Diversity}. It selects participants for one training task rather than forming multiple endogenous, stable coalitions. Model-sharing games study voluntary participation when participants may prefer models trained by selected groups \cite{DonahueKleinberg2021ModelSharing}. Game-theoretic analyses of the collaboration itself study equilibria of sampling contributions \cite{BlumEtAl2021Collaboration}, core-stable aggregation for fairness across heterogeneous data owners \cite{ChaudhuryEtAl2022Core}, and collaboration-structure formation among self-interested and possibly competing silos \cite{ChenEtAl2024FedEgoists}. These works allocate model quality or sampling burden; they do not design monetary transfers under a coordinator budget. Budget-feasible mechanism design in the procurement-auction sense bounds payments to strategic sellers by a hard budget \cite{Singer2010BudgetFeasible}; weak budget feasibility in this paper instead requires nonnegative coordinator-retained surplus for a posted allocation rule under complete information.

The coalition-specific model belongs to clustered federated learning. IFCA alternates participant-cluster assignment and cluster-model training \cite{GhoshEtAl2020IFCA}, while clustered federated learning based on loss-surface geometry groups participants with jointly trainable distributions \cite{SattlerEtAl2021CFL}. These methods address statistical clustering but do not by themselves determine monetary transfers or strategic participation.

Our analysis builds on hedonic coalition formation, especially existence and stability results for additively separable preferences \cite{BogomolnaiaJackson2002,Hajdukova2006,AzizBrandl2012,AzizSavani2016HedonicChapter}. Convergence of decentralized deviation dynamics to individually stable partitions has been characterized across many hedonic-game classes \cite{BrandtBullingerWilczynski2023}, and Nash stability and efficiency have been analyzed in depth for fractional hedonic games \cite{BiloEtAl2018FHG}; in our game, both properties follow from a single exact potential induced by the transfer design. Donahue and Kleinberg study optimal and stable FL coalition arrangements under an error-based hedonic model, give an efficient optimal-partition algorithm for their structure, and prove a price-of-anarchy bound of nine \cite{DonahueKleinberg2021Optimality}. Hedonic formulations have also been deployed operationally in federated systems, including reputation-aware coalition formation for serverless hierarchical FL \cite{NgEtAl2022ReputationHedonic}, a dual-level coalition-auction participation game \cite{ChenEtAl2025DualGFL}, and coalition formation with contract-based rewards in aerial IoT federations \cite{LiuEtAl2025AAV}. These designs adopt stability notions and algorithms tailored to their target architectures. None of them characterizes the allocation-rule class that guarantees stable partitions under a coordinator budget, and none connects retained budget slack to certified welfare bounds. Our differentiator is therefore not hedonic FL stability by itself. It is the coupling of participant preferences to affordable monetary transfers and the resulting relationship among participant utility, coordinator-retained surplus, and total welfare.

The pairwise potential objective is closely related to correlation clustering \cite{BansalBlumChawla2004,Swamy2004}. We use this relationship to state the computational limitation of global potential optimization and to transfer approximation guarantees to stable post-processed partitions.

\section{System Model}
\subsection{Participants and coalition-specific objectives}
Let $\N=\{1,\ldots,n\}$ denote the finite set of participants. Participant $i$ owns a local data set $\calD_i$ of size $m_i$ and has empirical loss
\begin{equation}
    F_i(\btheta)=\frac{1}{m_i}\sum_{z\in\calD_i}\ell(\btheta;z).
\end{equation}
A coalition $S\subseteq\N$, $S\neq\emptyset$, trains a coalition-specific model. Its weighted training objective is
\begin{equation}
    F_S(\btheta)=\sum_{i\in S}\alpha_i^S F_i(\btheta),
    \qquad
    \alpha_i^S=\frac{m_i}{\sum_{j\in S}m_j}.
    \label{eq:coalition-objective}
\end{equation}
A partition $\Pi\in\Part(\N)$ is a collection of nonempty, disjoint coalitions whose union is $\N$. Each $S\in\Pi$ trains and uses its own model $\btheta_S$. This interpretation makes the statistical outcome depend directly on the partition. \Cref{fig:coalition-partition} illustrates the resulting coalition-specific structure.

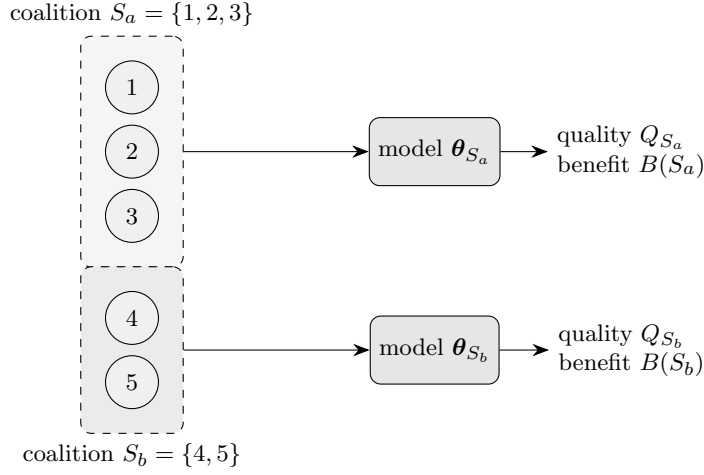
\begin{figure}[t]
\centering
\begin{tikzpicture}[font=\small,>={Stealth[length=2.2mm]},
  part/.style={circle,draw,minimum size=7mm,inner sep=1pt,fill=gray!12},
  model/.style={draw,rounded corners,minimum width=17mm,minimum height=9mm,fill=gray!20,align=center},
  val/.style={align=left}]
\node[part] (a1) at (0,0)     {$1$};
\node[part] (a2) at (0,-0.85) {$2$};
\node[part] (a3) at (0,-1.7)  {$3$};
\node[part] (b1) at (0,-3.05) {$4$};
\node[part] (b2) at (0,-3.9)  {$5$};
\begin{scope}[on background layer]
\node[draw,dashed,rounded corners,fill=black!4,fit=(a1)(a3),inner sep=3.2mm,
      label={[font=\small]above:{coalition $S_a=\{1,2,3\}$}}] (SA) {};
\node[draw,dashed,rounded corners,fill=black!8,fit=(b1)(b2),inner sep=3.2mm,
      label={[font=\small]below:{coalition $S_b=\{4,5\}$}}] (SB) {};
\end{scope}
\node[model] (mA) at (4,-0.85) {model $\btheta_{S_a}$};
\node[model] (mB) at (4,-3.475){model $\btheta_{S_b}$};
\node[val] (oA) at (6.6,-0.85) {quality $Q_{S_a}$\\ benefit $B(S_a)$};
\node[val] (oB) at (6.6,-3.475){quality $Q_{S_b}$\\ benefit $B(S_b)$};
\draw[->] (SA.east) -- (mA.west);
\draw[->] (SB.east) -- (mB.west);
\draw[->] (mA.east) -- (oA.west);
\draw[->] (mB.east) -- (oB.west);
\end{tikzpicture}
\caption{Coalition-specific training. A partition assigns each participant to exactly one coalition; every coalition $S$ trains its own model $\btheta_S$ from its members' weighted objective~\eqref{eq:coalition-objective} and produces terminal quality $Q_S$ and expected benefit $B(S)$. Because each coalition trains a different model, the statistical outcome depends on the partition.}
\label{fig:coalition-partition}
\end{figure}

\subsection{Communication and aggregation}
Every coalition starts from the same deterministic initialization,
\begin{equation}
    \btheta_S^0=\btheta^{\mathrm{init}},
    \qquad \emptyset\neq S\subseteq\N.
    \label{eq:initialization}
\end{equation}
Coalition-specific random initialization can instead be included as an additional coordinate of the probability space below.

Training proceeds over rounds $r=0,1,\ldots,R-1$. Let $X_i^r\in\{0,1\}$ indicate whether participant $i$ successfully returns its local update in round $r$. We initially assume
\begin{equation}
    \Pr(X_i^r=1)=p_i,
\end{equation}
with independence across participants and rounds. Let $\xi_i^r$ collect participant $i$'s local training randomness in round $r$, including minibatch sampling, data augmentation, and optimizer noise. For precision, let $(\Omega_{\xi_i^r},\mathcal F_{\xi_i^r},\mathbb P_{\xi_i^r})$ be its probability space. The baseline joint space is the finite product
\begin{equation}
    (\Omega,\mathcal F,\mathbb P)
    =\bigotimes_{r=0}^{R-1}\bigotimes_{i\in\N}
    \left[
    (\Omega_{\xi_i^r},\mathcal F_{\xi_i^r},\mathbb P_{\xi_i^r})
    \otimes
    (\{0,1\},2^{\{0,1\}},\operatorname{Bern}(p_i))
    \right].
    \label{eq:training-probability-space}
\end{equation}
Thus, local randomness is independent across participants and rounds and independent of communication success. Correlated failures or optimization noise can be incorporated by replacing the product measure with a specified joint law.

For coalition $S$, let
\begin{equation}
    A_S^r=\{i\in S:X_i^r=1\}
\end{equation}
be the set of successfully received updates. The baseline local algorithm is $K_i$ steps of minibatch stochastic gradient descent. If $\mathcal B_i^{r,k}$ is the minibatch encoded in $\xi_i^r$ and $\widehat F_i(\cdot;\mathcal B_i^{r,k})$ is its empirical loss, define
\begin{align}
    \btheta_{i\mid S}^{r,0}&=\btheta_S^r,\notag\\
    \btheta_{i\mid S}^{r,k+1}
    &=\btheta_{i\mid S}^{r,k}
      -\eta_i^{r,k}\nabla\widehat F_i
       (\btheta_{i\mid S}^{r,k};\mathcal B_i^{r,k}),
      &&k=0,\ldots,K_i-1,\notag\\
    \btheta_{i\mid S}^{r+1}&=\btheta_{i\mid S}^{r,K_i}.
    \label{eq:local-sgd}
\end{align}
The economic and coalition-formation results below depend on the induced coalition values, not on SGD specifically, so another measurable local solver can replace \eqref{eq:local-sgd} without changing those results.
The coordinator applies
\begin{equation}
\btheta_S^{r+1}=
\begin{cases}
\displaystyle
\frac{\sum_{i\in A_S^r}m_i\btheta_{i\mid S}^{r+1}}
     {\sum_{i\in A_S^r}m_i}, & A_S^r\neq\emptyset,\\[2.2ex]
\btheta_S^r, & A_S^r=\emptyset.
\end{cases}
\label{eq:safe-aggregation}
\end{equation}
The second case removes the zero-denominator event present in an unconditional weighted average. \Cref{fig:safe-aggregation-round} summarizes one communication round. We do not assume that merging two coalitions necessarily reduces loss. Such an inequality is generally false under heterogeneous data and nonconvex learning. Any later performance bound will state the required learning assumptions explicitly.

\begin{figure}[t]
\centering
\begin{tikzpicture}[font=\small,>={Stealth[length=2.2mm]},
  mem/.style={draw,rounded corners,minimum width=25mm,minimum height=10mm,align=center,fill=gray!10},
  coord/.style={draw,rounded corners,minimum width=44mm,minimum height=9mm,align=center,fill=gray!20},
  ok/.style={->,thick},
  fail/.style={->,dashed,red!75!black,thick}]
\node[coord] (bc) at (0,0) {coordinator broadcasts $\btheta_S^r$};
\node[mem] (m1) at (-4,-1.7) {member $1$\\ $K_1$ local SGD steps};
\node[mem] (m2) at (0,-1.7)  {member $2$\\ $K_2$ local SGD steps};
\node[mem] (m3) at (4,-1.7)  {member $3$\\ $K_3$ local SGD steps};
\draw[->] (bc.south) -- (m1.north);
\draw[->] (bc.south) -- (m2.north);
\draw[->] (bc.south) -- (m3.north);
\node[coord] (agg) at (0,-3.7) {coordinator: safe aggregation~\eqref{eq:safe-aggregation}};
\draw[ok]   (m1.south) -- node[left,font=\footnotesize]{$X_1^r=1$} (agg.north west);
\draw[fail] (m2.south) -- node[fill=white,inner sep=1pt,font=\footnotesize]{$X_2^r=0$ (dropped)} (agg.north);
\draw[ok]   (m3.south) -- node[right,font=\footnotesize]{$X_3^r=1$} (agg.north east);
\node[align=center,font=\footnotesize,below=3mm of agg]
   {aggregate only $A_S^r=\{i:X_i^r=1\}$, weighted by $m_i$;\quad
    if $A_S^r=\emptyset$ then $\btheta_S^{r+1}=\btheta_S^r$ (model unchanged)};
\end{tikzpicture}
\caption{One communication round for a coalition $S$. The coordinator broadcasts the current model $\btheta_S^r$; each member runs $K_i$ local SGD steps~\eqref{eq:local-sgd}; successful uploads ($X_i^r=1$) are averaged with weights $m_i$, while dropped uploads ($X_i^r=0$) are skipped, following the safe-aggregation rule~\eqref{eq:safe-aggregation}. If no update arrives ($A_S^r=\emptyset$), the model is left unchanged.}
\label{fig:safe-aggregation-round}
\end{figure}

\subsection{Learning benefit}
Let $Q_S(\btheta_S^R)$ be the terminal quality of coalition $S$ measured on a specified validation distribution. Higher $Q_S$ means better performance. Examples include negative validation loss, accuracy subject to a calibration constraint, or a task-specific monetary score. The expected gross benefit of $S$ is
\begin{equation}
    B(S)=\E\!\left[b_S\!\left(Q_S(\btheta_S^R)\right)\right],
    \label{eq:benefit}
\end{equation}
where the expectation is taken over the explicitly modeled training variables $\xi$ and communication variables $X$. Equation~\eqref{eq:benefit} deliberately separates statistical quality from its monetary valuation. It also permits empirical estimation when a closed-form learning bound is unavailable.

In the baseline value flow, the coordinator operates and monetizes a coalition-specific model service and receives benefit $B(S)$. Data-owning participants supply local updates, receive access to the resulting service when relevant, and are paid transfers defined below. A member-owned cooperative interpretation would require a different allocation of the gross benefit and is not assumed here.

\section{Economic and Hedonic Model}
\subsection{Costs, transfers, and transferable surplus}
Let $C_0(S)$ denote the coordinator's non-transfer operating cost for coalition $S$, including communication, aggregation, and model-serving costs. Let $d_i(S)$ denote participant $i$'s cost of belonging to $S$, including computation, energy, communication, delay, and privacy cost. The total transferable surplus created by coalition $S$ is
\begin{equation}
    W(S)=B(S)-C_0(S)-\sum_{i\in S}d_i(S).
    \label{eq:surplus}
\end{equation}

\begin{assumption}[Pre-screened participant population]
The coordinator admits only individually viable participants, so
\begin{equation}
    W(\{i\})\geq0,\qquad i\in\N.
    \label{eq:singleton-viability}
\end{equation}
Admission control for participants who fail \eqref{eq:singleton-viability} is outside the baseline model.
\end{assumption}

The coordinator pays transfer $t_i(S)$ to participant $i\in S$. Participant utility and retained coordinator surplus are, respectively,
\begin{align}
    U_i(S)&=t_i(S)-d_i(S),
    \label{eq:agent-utility}\\
    R_0(S)&=B(S)-C_0(S)-\sum_{i\in S}t_i(S)
           =W(S)-\sum_{i\in S}U_i(S).
    \label{eq:retained-surplus}
\end{align}
We normalize the utility from operating alone to $U_i(\{i\})=0$. A nonzero outside option can be included by subtracting it from all coalition utilities involving $i$.

\begin{definition}[Weak budget feasibility]
An allocation rule $U=(U_i(S))_{S\subseteq\N,i\in S}$ is weakly budget feasible if
\begin{equation}
    \sum_{i\in S}U_i(S)\leq W(S),
    \qquad \emptyset\neq S\subseteq\N.
    \label{eq:budget-feasibility}
\end{equation}
Equivalently, the coordinator retains nonnegative surplus in every coalition that forms.
\end{definition}

Exact budget balance replaces the inequality in \eqref{eq:budget-feasibility} with equality. Weak budget feasibility is used because the coordinator may retain part of the surplus. Individual rationality in a realized partition requires $U_i(S)\geq0$ for every $i\in S$. We will show that it follows from Nash stability when operating alone is an available deviation.

\subsection{Hedonic preferences and Nash stability}
The allocation rule induces participant preferences. For coalitions $S,T\subseteq\N$ containing $i$,
\begin{equation}
    S\succeq_i T
    \quad\Longleftrightarrow\quad
    U_i(S)\geq U_i(T).
\end{equation}
Preferences are hedonic because $U_i(S)$ depends only on the members of $S$ and not on coalitions elsewhere in the partition. This assumption is appropriate when coalition models, benefits, costs, and transfers are operationally independent. Shared capacity constraints or cross-coalition competition would create externalities and require a partition-function game.

For $i\in\N$, let $S_\Pi(i)$ be the unique coalition in $\Pi$ containing $i$. A move by $i$ from $S_\Pi(i)$ to $T\in\Pi\cup\{\emptyset\}$ produces the coalition $T\cup\{i\}$.

\begin{definition}[Nash-stable partition]
A partition $\PiNS$ is Nash stable under allocation rule $U$ if
\begin{equation}
    U_i\bigl(S_{\PiNS}(i)\bigr)
    \geq U_i(T\cup\{i\}),
    \quad
    i\in\N,\quad T\in\PiNS\cup\{\emptyset\}.
    \label{eq:nash-stability}
\end{equation}
\end{definition}

Nash stability assumes that a destination coalition cannot reject an entrant. In federated learning, however, an entrant changes the coalition model and may harm existing members. We therefore also use the classical acceptance-based concept of individual stability \cite{DrezeGreenberg1980,BogomolnaiaJackson2002}.

\begin{definition}[Individual stability]
A move by $i$ from $S_\Pi(i)$ to $T\in\Pi\cup\{\emptyset\}$ is individually admissible if
\begin{align}
    U_i(T\cup\{i\})&>U_i(S_\Pi(i)),
    \label{eq:is-profitable}\\
    U_j(T\cup\{i\})&\geq U_j(T),
    &&j\in T.
    \label{eq:is-accepted}
\end{align}
A partition is individually stable if it admits no individually admissible move.
\label{def:individual-stability}
\end{definition}

Individual stability gives every destination member a veto but does not require consent from members left behind. Contractual individual stability would add departure consent and is outside the present scope.

\begin{proposition}[Stability implies participant rationality]
If $U_i(\{i\})=0$ for every $i$, then every Nash-stable partition is individually rational.
\end{proposition}
\begin{proof}
Choose $T=\emptyset$ in \eqref{eq:nash-stability}. Then
$U_i(S_{\PiNS}(i))\geq U_i(\{i\})=0$ for every participant.
\end{proof}

\section{Symmetric Pairwise Surplus Allocation}
\subsection{Allocation class}
Let $v_{ij}=v_{ji}\in\R$ be the pairwise surplus assigned to each of participants $i$ and $j$ when they belong to the same coalition, and let $v_{ii}=0$. The induced participant utility is
\begin{equation}
    U_i^{v}(S)=\sum_{j\in S\setminus\{i\}}v_{ij}.
    \label{eq:additive-utility}
\end{equation}
Negative values represent incompatibility, for example statistical interference, delay, or a costly collaboration relationship. Positive values represent beneficial complementarity.

The pairwise restriction applies to the transfer-induced utilities, not to the physical coalition surplus $W(S)$. The learning benefit and costs may contain asymmetric and higher-order effects. Each endpoint receives the full value $v_{ij}$, which explains the factor of two below. A rule that splits a single pair value between the endpoints would use different notation and budget accounting.

For $S\subseteq\N$, define the unordered edge set
\begin{equation}
    E(S)=\bigl\{\{i,j\}:i,j\in S,\ i<j\bigr\}.
\end{equation}
Because every pair value is assigned to both endpoints,
\begin{equation}
    \sum_{i\in S}U_i^v(S)=2\sum_{\{i,j\}\in E(S)}v_{ij}.
\end{equation}
Thus, the pairwise allocation is weakly budget feasible if
\begin{equation}
    2\sum_{\{i,j\}\in E(S)}v_{ij}\leq W(S),
    \qquad \emptyset\neq S\subseteq\N.
    \label{eq:pair-budget}
\end{equation}

To exclude economically meaningless unbounded negative pair values, the design model may impose compatibility bounds
\begin{equation}
    \underline v_{ij}\leq v_{ij}\leq\overline v_{ij}.
    \label{eq:pair-bounds}
\end{equation}
The bounds can be estimated from pairwise validation effects, cost differences, or explicit policy constraints. Define the feasible allocation polytope
\begin{equation}
\calV(W)=\left\{v:
\eqref{eq:pair-budget}\text{ and }\eqref{eq:pair-bounds}\text{ hold}
\right\}.
\label{eq:allocation-polytope}
\end{equation}
The following result gives an exact feasibility test. Its direct form still contains one condition for every coalition, which motivates the later search for useful structural reductions.

\begin{theorem}[Feasibility of bounded pairwise allocations]
The polytope $\calV(W)$ is nonempty if and only if
\begin{align}
    \underline v_{ij}&\leq\overline v_{ij}, &&i<j,
    \label{eq:ordered-bounds}\\
    2\sum_{\{i,j\}\in E(S)}\underline v_{ij}&\leq W(S),
    &&\emptyset\neq S\subseteq\N.
    \label{eq:lower-feasibility}
\end{align}
\label{thm:pair-feasibility}
\end{theorem}
\begin{proof}
Necessity of \eqref{eq:ordered-bounds} follows from the box constraints. For every feasible $v$, the componentwise inequalities $v_{ij}\geq\underline v_{ij}$ and the nonnegative coefficients in \eqref{eq:pair-budget} imply
\[
2\sum_{\{i,j\}\in E(S)}\underline v_{ij}
\leq
2\sum_{\{i,j\}\in E(S)}v_{ij}
\leq W(S),
\]
which proves necessity of \eqref{eq:lower-feasibility}. Conversely, if \eqref{eq:ordered-bounds} and \eqref{eq:lower-feasibility} hold, setting $v_{ij}=\underline v_{ij}$ satisfies every box and coalition-budget constraint. Hence $v\in\calV(W)$.
\end{proof}

For a singleton $S=\{i\}$, condition \eqref{eq:lower-feasibility} reduces to $0\leq W(\{i\})$, which is exactly the pre-screening condition \eqref{eq:singleton-viability}.

\begin{corollary}[Nonnegative pair rewards]
If $\underline v_{ij}=0$ and $\overline v_{ij}\geq0$ for all $i<j$, then $\calV(W)$ is nonempty if and only if $W(S)\geq0$ for every nonempty coalition $S$.
\end{corollary}

The zero allocation is feasible under the corollary but makes every participant indifferent among all coalitions. Feasibility alone is therefore not a sufficient design objective. Positive lower bounds, compatibility targets, or the objective in \eqref{eq:allocation-lp} are needed to create informative incentives.

\subsection{Exact potential and stable-partition existence}
For a partition $\Pi$, define
\begin{equation}
    P_v(\Pi)=\sum_{S\in\Pi}\sum_{\{i,j\}\in E(S)}v_{ij}.
    \label{eq:potential}
\end{equation}
Equivalently,
\begin{equation}
    P_v(\Pi)=\frac{1}{2}\sum_{i\in\N}U_i^v(S_\Pi(i)).
    \label{eq:potential-participant-utility}
\end{equation}
Exact potential games are due to Monderer and Shapley \cite{MondererShapley1996}: every unilateral deviation changes the deviator's utility and the potential by exactly the same amount.

\begin{theorem}[Exact potential]
For every symmetric pairwise allocation $v$, the hedonic game induced by \eqref{eq:additive-utility} is an exact potential game with potential \eqref{eq:potential}.
\label{thm:exact-potential}
\end{theorem}
\begin{proof}
Consider a unilateral move by participant $i$ from coalition $S$ to coalition $T$, where $i\in S$ and $i\notin T$. Terms in \eqref{eq:potential} that do not involve $i$ remain unchanged. Therefore,
\begin{align}
P_v(\Pi')-P_v(\Pi)
&=\sum_{j\in T}v_{ij}-\sum_{j\in S\setminus\{i\}}v_{ij}\\
&=U_i^v(T\cup\{i\})-U_i^v(S).
\end{align}
The change in the potential equals the deviating participant's utility change, which proves the exact-potential property.
\end{proof}

The existence of a potential for symmetric additively separable hedonic preferences is classical \cite{BogomolnaiaJackson2002}. The explicit proof is retained because the same potential enters the budget and welfare analysis below.

\begin{corollary}[Existence]
Every finite symmetric pairwise hedonic game admits at least one Nash-stable partition.
\label{cor:existence}
\end{corollary}
\begin{proof}
The finite set $\Part(\N)$ contains a maximizer of $P_v$. If a participant had a strictly profitable unilateral deviation from such a maximizer, \Cref{thm:exact-potential} would imply a strict increase in $P_v$, contradicting maximality.
\end{proof}

\begin{corollary}[Finite improvement property]
Every sequence of strict unilateral utility improvements terminates at a Nash-stable partition.
\label{cor:fip}
\end{corollary}
\begin{proof}
Each strict improvement strictly increases $P_v$ by \Cref{thm:exact-potential}. Since the set of partitions is finite, no partition can repeat and the sequence must terminate. At termination, no strict profitable unilateral move remains.
\end{proof}

\begin{proposition}[Destination consent and individual stability]
Every Nash-stable partition is individually stable. Under symmetric pairwise utilities, a profitable move by $i$ to $T$ is accepted by all destination members if and only if $v_{ij}\geq0$ for every $j\in T$. From any initial partition, every sequence of individually admissible moves terminates at an individually stable partition.
\label{prop:individual-stability}
\end{proposition}
\begin{proof}
Every individually admissible move is a profitable unilateral move, so Nash stability rules it out. Under \eqref{eq:additive-utility}, destination member $j$ experiences the utility change
\[
    U_j^v(T\cup\{i\})-U_j^v(T)=v_{ij},
\]
which proves the acceptance characterization. Every individually admissible move strictly raises the deviator's utility and therefore strictly raises $P_v$ by \Cref{thm:exact-potential}. Finiteness of $\Part(\N)$ implies termination. At termination no individually admissible move remains, which is individual stability.
\end{proof}

\begin{remark}[Local and global optimality]
Every global maximizer of $P_v$ is Nash stable. A Nash-stable partition need only be a local maximum with respect to unilateral moves and need not maximize $P_v$ globally. Stability, potential optimality, and social-welfare optimality are therefore distinct concepts.
\end{remark}

\subsection{Coordinator slack and welfare efficiency}
For any pairwise allocation $v$, define coalition-level budget slack and partition-level slack by
\begin{align}
    r_v(S)&=W(S)-2\sum_{\{i,j\}\in E(S)}v_{ij},
    \label{eq:coalition-slack}\\
    R_v(\Pi)&=\sum_{S\in\Pi}r_v(S).
    \label{eq:partition-slack}
\end{align}
Equation~\eqref{eq:retained-surplus} shows unconditionally that $r_v(S)$ is the coordinator surplus retained from coalition $S$ under the pairwise utility rule. Budget feasibility is equivalent to $r_v(S)\geq0$ for every nonempty coalition. Define total social welfare as
\begin{equation}
    \operatorname{SW}(\Pi)=\sum_{S\in\Pi}W(S).
    \label{eq:social-welfare}
\end{equation}
The welfare accounting decomposes as
\begin{equation}
    \operatorname{SW}(\Pi)=2P_v(\Pi)+R_v(\Pi).
    \label{eq:welfare-decomposition}
\end{equation}
Thus, the potential measures half of total participant utility, while $R_v$ measures the remaining welfare retained by the coordinator.

\begin{theorem}[Additive welfare guarantee from optimal-partition slack]
Let $v$ be a weakly budget-feasible symmetric pairwise allocation, let $\Pi^P\in\arg\max_{\Pi\in\Part(\N)}P_v(\Pi)$, and let $\Pi^\star\in\arg\max_{\Pi\in\Part(\N)}\operatorname{SW}(\Pi)$. Then $\Pi^P$ is Nash stable and
\begin{equation}
    \operatorname{SW}(\Pi^P)
    \geq
    \operatorname{SW}(\Pi^\star)-R_v(\Pi^\star).
    \label{eq:additive-welfare-bound}
\end{equation}
\label{thm:slack-welfare}
\end{theorem}
\begin{proof}
Nash stability of $\Pi^P$ follows from \Cref{thm:exact-potential}. Using \eqref{eq:welfare-decomposition}, optimality of $\Pi^P$ for $P_v$, and nonnegativity of $R_v(\Pi^P)$,
\begin{align*}
\operatorname{SW}(\Pi^\star)
&=2P_v(\Pi^\star)+R_v(\Pi^\star)\\
&\leq2P_v(\Pi^P)+R_v(\Pi^\star)\\
&\leq\operatorname{SW}(\Pi^P)+R_v(\Pi^\star).
\end{align*}
Rearranging proves \eqref{eq:additive-welfare-bound}.
\end{proof}

This is an existence and price-of-stability statement, in the standard sense of comparing the best stable outcome with the optimum \cite{AnshelevichEtAl2008PoS}. Computing the global potential maximizer $\Pi^P$ is NP-hard in general, as discussed after \Cref{prop:correlation-equivalence}. \Cref{thm:approximate-stable-welfare} below gives an objective-value guarantee for an approximation followed by stabilization.

\begin{corollary}[Exact budget balance]
If $r_v(S)=0$ for every nonempty coalition $S$, then a welfare-optimal Nash-stable partition exists. Whenever the welfare ratio is defined, the price of stability is one.
\label{cor:exact-balance-pos}
\end{corollary}
\begin{proof}
Exact balance gives $R_v(\Pi^\star)=0$. Apply \Cref{thm:slack-welfare}.
\end{proof}

\begin{remark}[Applicability of exact balance]
Exact balance in every coalition is achievable only when there are pair values satisfying
\[
    W(S)=2\sum_{\{i,j\}\in E(S)}v_{ij}
    \qquad\text{for every nonempty }S.
\]
Thus, it requires exact pairwise representability of the physical surplus and forces $W(\{i\})=0$ for all singletons. This is a restrictive joint condition on $W$ and $v$, not a transfer choice available for arbitrary clustered-FL surplus. The uniform-slack result below covers a near-representable case.
\end{remark}

\begin{corollary}[Uniform coalition slack]
If $v$ is weakly budget feasible and $r_v(S)\leq\varepsilon$ for every nonempty coalition $S$, then there exists a Nash-stable partition $\Pi^P$ satisfying
\begin{equation}
    \operatorname{SW}(\Pi^P)
    \geq
    \operatorname{SW}(\Pi^\star)-n\varepsilon.
\end{equation}
\end{corollary}
\begin{proof}
The welfare-optimal partition contains at most $n$ coalitions, so $R_v(\Pi^\star)\leq n\varepsilon$. Apply \Cref{thm:slack-welfare}.
\end{proof}

\begin{corollary}[Relative slack and price of stability]
Under the hypotheses of \Cref{thm:slack-welfare}, suppose $\operatorname{SW}(\Pi^\star)>0$ and
\begin{equation}
    R_v(\Pi^\star)
    \leq\delta\operatorname{SW}(\Pi^\star),
    \qquad 0\leq\delta<1.
    \label{eq:relative-slack}
\end{equation}
Then a Nash-stable partition $\Pi^P$ satisfies
\begin{equation}
    \operatorname{SW}(\Pi^P)
    \geq(1-\delta)\operatorname{SW}(\Pi^\star),
\end{equation}
and the price of stability is at most $1/(1-\delta)$. A sufficient coalition-level condition is
\begin{equation}
    0\leq r_v(S)\leq\delta W(S)
    \qquad\text{for every nonempty }S.
    \label{eq:coalition-relative-slack}
\end{equation}
\label{cor:relative-slack}
\end{corollary}
\begin{proof}
The welfare bound follows by substituting \eqref{eq:relative-slack} into \eqref{eq:additive-welfare-bound}. Since $\Pi^P$ is stable, the best stable welfare is at least $\operatorname{SW}(\Pi^P)$, which gives the price-of-stability bound. Summing \eqref{eq:coalition-relative-slack} over the coalitions in $\Pi^\star$ gives \eqref{eq:relative-slack}.
\end{proof}

\begin{proposition}[Asymptotic tightness of the relative-slack bound]
For every $\delta\in(0,1)$ and $\epsilon\in(0,\delta/2)$, there is a three-participant weakly budget-feasible instance with a unique welfare-optimal partition $\Pi^\star$ and a unique Nash-stable partition $\Pi^G$ such that
\[
    \operatorname{SW}(\Pi^\star)=1,
    \qquad
    R_v(\Pi^\star)=\delta,
    \qquad
    \operatorname{SW}(\Pi^G)=1-\delta+2\epsilon.
\]
Consequently, as $\epsilon$ tends to zero, the price of stability approaches $1/(1-\delta)$.
\label{prop:relative-slack-tightness}
\end{proposition}
\begin{proof}
Let $\N=\{1,2,3\}$, put $a=(1-\delta)/2$, and set
\[
    v_{12}=a,
    \qquad
    v_{13}=v_{23}=\epsilon/2.
\]
Define the coalition surplus by
\begin{align*}
    &W(\{1\})=W(\{2\})=0,
    &&W(\{3\})=\delta,\\
    &W(\{1,2\})=1-\delta,
    &&W(\{1,3\})=W(\{2,3\})=\epsilon,\\
    &W(\N)=1-\delta+2\epsilon.
\end{align*}
Every budget constraint holds with equality except that the singleton $\{3\}$ retains slack $\delta$. The partition $\Pi^\star=\{\{1,2\},\{3\}\}$ has welfare one and slack $\delta$. The other partition welfares are $1-\delta+2\epsilon$ for the grand coalition, $\epsilon$ for either remaining pair-plus-singleton partition, and $\delta$ for the all-singleton partition. The restrictions on $\delta$ and $\epsilon$ therefore make $\Pi^\star$ uniquely welfare optimal. All three pair values are strictly positive. Hence every nongrand partition admits a strict improvement by a singleton participant joining another coalition, while no participant in the grand coalition benefits from leaving to a singleton. Thus $\Pi^G=\{\N\}$ is uniquely Nash stable, and its welfare has the stated value.
\end{proof}

This construction establishes tightness for the partition-level hypothesis \eqref{eq:relative-slack}. It does not satisfy the stronger coalition-level sufficient condition \eqref{eq:coalition-relative-slack}; whether that condition permits a better worst-case guarantee remains open.

Budget feasibility alone does not bound efficiency because it restricts participant utility only from above.

\begin{proposition}[Unbounded welfare loss under budget feasibility alone]
For every $\rho>1$, there exists a two-participant instance with nonnegative coalition surplus and a weakly budget-feasible symmetric pairwise allocation whose unique Nash-stable partition has welfare $1$, while the welfare-optimal partition has welfare $\rho$.
\label{prop:unbounded-welfare}
\end{proposition}
\begin{proof}
Let $\N=\{1,2\}$, set $W(\{1\})=W(\{2\})=1/2$ and $W(\{1,2\})=\rho$, and choose $v_{12}=-1$. The singleton constraints satisfy $0\leq1/2$, and the pair constraint satisfies $2v_{12}=-2\leq\rho$, so the allocation is weakly budget feasible. In the grand coalition each participant obtains utility $-1$ and strictly prefers its singleton utility $0$. In the singleton partition, joining the other participant would reduce utility to $-1$. Hence the singleton partition is uniquely Nash stable and has welfare $1$, whereas the grand coalition is welfare optimal with welfare $\rho$.
\end{proof}

\Cref{prop:unbounded-welfare} shows that a finite price-of-stability or price-of-anarchy bound requires an alignment condition between transfers and primitive coalition surplus. The retained-slack conditions in \Cref{thm:slack-welfare,cor:exact-balance-pos} provide one such alignment.

\section{Incentive and Coalition Design Problems}
\subsection{Pairwise incentive design}
A basic allocation problem chooses pair values that reward desired compatibility while respecting every coalition's budget. Let $\omega_{ij}\geq0$ be design weights. We consider
\begin{align}
    \max_{v}\quad &\sum_{i<j}\omega_{ij}v_{ij}
    \label{eq:allocation-lp}\\
    \text{subject to}\quad
    &2\sum_{\{i,j\}\in E(S)}v_{ij}\leq W(S),
      &&\emptyset\neq S\subseteq\N,\notag\\
    &\underline v_{ij}\leq v_{ij}\leq\overline v_{ij},
      &&i<j.\notag
\end{align}
Unlike the conference formulation, \eqref{eq:allocation-lp} explicitly bounds pair values and does not treat feasibility itself as a nonemptiness theorem. \Cref{thm:pair-feasibility} characterizes feasibility, but both its direct test and \eqref{eq:allocation-lp} contain exponentially many coalition constraints.

There is also an exponential input problem: an unrestricted representation of $W$ requires $2^n-1$ coalition values, and estimating each $B(S)$ may require training and validating a coalition-specific model. A scalable implementation must therefore use a structured surplus model, a restricted family of reachable coalitions, or an estimator that answers coalition-value queries without training every coalition separately.

One useful structure gives a polynomial separation test. Extend $W$ by $W(\emptyset)=0$ and define $r_v(\emptyset)=0$. Recall that a set function $f:2^\N\to\R$ is submodular if
\begin{equation}
    f(A)+f(B)\geq f(A\cup B)+f(A\cap B),
    \qquad A,B\subseteq\N.
    \label{eq:submodularity}
\end{equation}

\begin{proposition}[Submodular verification of budget feasibility]
For a fixed pair allocation $v$, the coalition budget constraints hold if and only if
\begin{equation}
    \min_{S\subseteq\N}r_v(S)\geq0.
    \label{eq:slack-minimization}
\end{equation}
If $r_v$ is submodular and can be evaluated in polynomial time, budget feasibility can be verified in polynomial oracle time. In particular, if $W$ is submodular and $v_{ij}\geq0$ for every pair, then $r_v$ is submodular.
\label{prop:submodular-verification}
\end{proposition}
\begin{proof}
The first statement follows directly from \eqref{eq:coalition-slack}, with the empty set contributing zero. A submodular set function can be minimized in strongly polynomial oracle time \cite{Schrijver2000Submodular}. For the sufficient condition, the internal-edge function
\[
    g_v(S)=2\sum_{\{i,j\}\in E(S)}v_{ij}
\]
is supermodular when all $v_{ij}\geq0$: the marginal contribution of adding one endpoint of an edge weakly increases when the other endpoint is already present. Hence $-g_v$ is submodular, and $r_v=W-g_v$ is the sum of two submodular functions.
\end{proof}

\begin{corollary}[Polynomial feasibility test for nonnegative lower bounds]
Suppose $W$ is submodular and available through a polynomial-time value oracle, and $0\leq\underline v_{ij}\leq\overline v_{ij}$ for every pair. Then nonemptiness of $\calV(W)$ can be decided in polynomial oracle time by minimizing $r_{\underline v}$ and applying \Cref{thm:pair-feasibility}.
\end{corollary}

The result is an oracle-complexity statement. It removes explicit enumeration of all coalition constraints but does not remove the statistical cost of constructing or querying a valid surplus model. For nonnegative pair bounds, the same minimization supplies a violated-coalition separation oracle for \eqref{eq:allocation-lp}. Under the usual rational-encoding assumptions, standard separation-based linear optimization then gives a polynomial oracle-time solution of the allocation LP, provided that $W$ itself has a polynomial-time value oracle.

Submodular $W$ is plausible in data-redundant regimes where benefit is concave in an additive effective-sample-size statistic and operating and participant costs are modular. It can fail in the complementary regimes that motivate clustering, for example when adding a participant completes missing label or feature coverage and creates increasing returns. The exact requirement in \Cref{prop:submodular-verification} is submodularity of $r_v$; the sufficient condition on $W$ is therefore a tractability frontier rather than a universal assumption about federated learning.

\subsection{Stable welfare maximization}
For a fixed allocation $v$, a system-level objective is
\begin{align}
    \max_{\Pi\in\Part(\N)}\quad
        &\sum_{S\in\Pi}W(S)
        \label{eq:stable-welfare}\\
    \text{subject to}\quad
        &\Pi\text{ is Nash stable under }U^v.\notag
\end{align}
This formulation directly connects the stable partition to learning benefit and economic cost. It replaces the earlier objective of minimizing the coalition value itself. \Cref{thm:slack-welfare} bounds the loss of the best potential-maximizing stable partition when coordinator-retained surplus is bounded. Joint design chooses both $v\in\calV(W)$ and a stable $\Pi$, producing a bilevel or mixed-integer optimization problem.

The journal version will study the following questions:
\begin{enumerate}[leftmargin=*]
    \item Beyond the submodular-slack case, which structured surplus estimators admit practical separation or reduced constraint families?
    \item Which FL-specific learning and cost assumptions imply a small relative-slack factor $\delta$ or small negative-edge mass $C_v$?
    \item Can the approximation-and-stabilization pipeline be given a polynomial convergence bound without enumerating all coalition values?
\end{enumerate}

\subsection{Correlation-clustering representation}
For a partition $\Pi$, define its weighted agreement objective as
\begin{equation}
\begin{split}
    A_v(\Pi)
    ={}&\sum_{\substack{i<j:\ S_\Pi(i)=S_\Pi(j),\ v_{ij}\geq0}}v_{ij}
    +\sum_{\substack{i<j:\ S_\Pi(i)\neq S_\Pi(j),\ v_{ij}<0}}|v_{ij}|.
\end{split}
\label{eq:max-agreement}
\end{equation}
This is the weighted maximum-agreement correlation-clustering objective: positive edges are rewarded inside coalitions and negative edges are rewarded across coalitions \cite{BansalBlumChawla2004,Swamy2004}.

\begin{proposition}[Correlation-clustering equivalence]
Let
\begin{equation}
    C_v=\sum_{\substack{i<j:\ v_{ij}<0}}|v_{ij}|.
\end{equation}
For every partition $\Pi$,
\begin{equation}
    A_v(\Pi)=C_v+P_v(\Pi).
    \label{eq:agreement-potential}
\end{equation}
Consequently, global maximization of $P_v$ and weighted maximum-agreement correlation clustering have the same optimal partitions.
\label{prop:correlation-equivalence}
\end{proposition}
\begin{proof}
The contribution of positive edges inside coalitions is identical in $A_v$ and $P_v$. For negative edges, $C_v$ counts the absolute weight of every negative edge. Subtracting the absolute weight of negative edges placed inside coalitions leaves exactly the absolute weight of negative edges placed across coalitions. Since an internal negative edge contributes its negative weight to $P_v$, the identity follows.
\end{proof}

The NP-hardness of maximum utilitarian welfare in symmetric additively separable hedonic games is already known \cite{AzizBrandtSeedig2013}. The equivalence above gives an alternative representation of the same optimization barrier. Its role here is to import correlation-clustering approximation algorithms \cite{BansalBlumChawla2004,Swamy2004} and connect them to stable post-processing.

\begin{proposition}[Stability-preserving post-processing]
Starting from any partition $\Pi^0$, apply arbitrary strict better-response moves until none remains. The process terminates at a Nash-stable partition $\widehat\Pi$ satisfying
\begin{equation}
    A_v(\widehat\Pi)\geq A_v(\Pi^0).
\end{equation}
If $\Pi^0$ is an $\alpha$-approximation for weighted maximum agreement, then $\widehat\Pi$ retains the same approximation guarantee.
\label{prop:stable-postprocessing}
\end{proposition}
\begin{proof}
Termination and Nash stability follow from \Cref{cor:fip}. Every strict better response increases $P_v$ by \Cref{thm:exact-potential}, and therefore increases $A_v$ by the constant-shift identity \eqref{eq:agreement-potential}. The final objective is at least the initial objective, so any multiplicative approximation guarantee is preserved.
\end{proof}

\begin{theorem}[Welfare guarantee for approximation and stabilization]
Let $v$ be weakly budget feasible, let
\[
    P_v^{\max}=\max_{\Pi\in\Part(\N)}P_v(\Pi),
\]
and let $\Pi^\star$ maximize social welfare. Suppose $\Pi^0$ is an $\alpha$-approximation for weighted maximum agreement, where $\alpha\in[0,1]$, and let $\widehat\Pi$ be the Nash-stable partition obtained by the post-processing in \Cref{prop:stable-postprocessing}. Then
\begin{equation}
    \operatorname{SW}(\widehat\Pi)
    \geq
    \alpha\operatorname{SW}(\Pi^\star)
    -\alpha R_v(\Pi^\star)
    -2(1-\alpha)C_v.
    \label{eq:approximate-stable-welfare}
\end{equation}
\label{thm:approximate-stable-welfare}
\end{theorem}
\begin{proof}
By \Cref{prop:stable-postprocessing,prop:correlation-equivalence},
\begin{align*}
    C_v+P_v(\widehat\Pi)
    &=A_v(\widehat\Pi)\\
    &\geq\alpha\max_{\Pi\in\Part(\N)}A_v(\Pi)
     =\alpha(C_v+P_v^{\max}),
\end{align*}
and hence
\[
    P_v(\widehat\Pi)
    \geq\alpha P_v^{\max}-(1-\alpha)C_v.
\]
Budget feasibility gives $R_v(\widehat\Pi)\geq0$. Since $P_v^{\max}\geq P_v(\Pi^\star)$, the welfare decomposition yields
\begin{align*}
    \operatorname{SW}(\widehat\Pi)
    &\geq2P_v(\widehat\Pi)\\
    &\geq2\alpha P_v(\Pi^\star)-2(1-\alpha)C_v\\
    &=\alpha\operatorname{SW}(\Pi^\star)
      -\alpha R_v(\Pi^\star)-2(1-\alpha)C_v.
\end{align*}
\end{proof}

\Cref{thm:approximate-stable-welfare} separates loss caused by coordinator-retained slack at the welfare optimum from loss caused by negative pair mass. The bound can be weak when either term is large, which makes both quantities important empirical diagnostics. The strict improvement phase can still require exponentially many moves. The post-processing and welfare results are objective guarantees, not polynomial-time guarantees for exact Nash stability.

\begin{proposition}[Tightness of the negative-edge correction]
For every $b>0$ and $a\geq2b$, there is an exactly budget-balanced four-participant instance in which a Nash-stable partition is an $\alpha$-approximation for weighted maximum agreement, with $\alpha=a/(a+2b)$, and equality holds in \eqref{eq:approximate-stable-welfare}. As $a$ and $b$ vary, the construction attains every $\alpha\in[1/2,1)$.
\label{prop:negative-edge-tightness}
\end{proposition}
\begin{proof}
Let $\N=\{1,2,3,4\}$, set $v_{12}=v_{34}=a$, and give each of the four cross-group pairs in $\{1,2\}\times\{3,4\}$ value $-b$. Define
\[
    W(S)=2\sum_{\{i,j\}\in E(S)}v_{ij}
    \qquad\text{for every nonempty }S.
\]
Thus $R_v(\Pi)=0$ for every partition. The partition $\Pi^\star=\{\{1,2\},\{3,4\}\}$ attains potential $2a$, the sum of all positive edge weights, and is therefore globally potential and welfare optimal. It also attains the maximum possible agreement value $2a+4b$. The grand coalition $\Pi^G=\{\N\}$ has potential $2a-4b$ and agreement value $2a$, so its exact agreement ratio is
\[
    \alpha=\frac{2a}{2a+4b}=\frac{a}{a+2b}.
\]
Every participant's utility in $\Pi^G$ is $a-2b\geq0$. A unilateral departure can only create a singleton with utility zero, so $\Pi^G$ is Nash stable and stabilization can leave it unchanged. Since $C_v=4b$, exact balance gives
\begin{align*}
    \operatorname{SW}(\Pi^G)
    &=4a-8b\\
    &=\alpha\operatorname{SW}(\Pi^\star)
      -2(1-\alpha)C_v,
\end{align*}
which is equality in \eqref{eq:approximate-stable-welfare} because $R_v(\Pi^\star)=0$.
\end{proof}

The construction uses net-negative surplus for some incompatible cross-group coalitions, which is permitted by the net-surplus model. It establishes worst-case tightness under the stated assumptions; additional restrictions such as nonnegative surplus for every coalition may yield stronger guarantees.

\section{Decentralized Coalition Formation}
Give every participant a label in $\Sigma=\{1,\ldots,n\}$. Participants with the same label form a coalition. A label profile $\bm{\sigma}\in\Sigma^n$ induces the partition
\begin{equation}
    \Pi(\bm{\sigma})=
    \left\{
    \{j\in\N:\sigma_j=s\}:s\in\Sigma
    \right\}\setminus\{\emptyset\}.
\end{equation}
Unused labels allow a participant to form a singleton coalition.

Starting from any label profile, a strict better-response procedure repeatedly selects a participant and moves it to a label that strictly raises \eqref{eq:additive-utility}. By \Cref{cor:fip}, the procedure terminates in finitely many moves. The result can depend on the initial partition and update order because different local maxima of $P_v$ may exist.

If destination consent is required, the move is executed only when $v_{ij}\geq0$ for every member $j$ of the destination coalition. By \Cref{prop:individual-stability}, this accepted better-response process terminates at an individually stable partition. The Nash procedure gives the stronger no-veto benchmark, while the accepted procedure models voluntary admission.

For implementation, participant $i$ needs the sum of its pair values with the members of each current coalition. The baseline is therefore semi-decentralized: the coordinator publishes current coalition memberships and privately or publicly supplies each participant's relevant pair scores. Given this information, all destination utilities can be evaluated in $O(n)$ arithmetic operations per participant. This observation does not bound the number of improvement steps, which may be exponential in the worst case. Establishing convergence rates under structured pair values is left for the full analysis.

\section{Empirical Evaluation}
\label{sec:empirical}

\subsection{Experimental design and provenance}

The empirical study tests whether the proposed allocation and decentralized adjustment procedure produces affordable, stable, and welfare-efficient coalitions when coalition values are estimated from federated learning runs. CIFAR-10 is the sole reported task. An earlier MNIST pilot was used only to validate the implementation and calibrate the protocol; because it was exploratory, its outcomes are excluded from the comparisons below. We preregistered the CIFAR-10 protocol before the substantive runs and retained all five analysis seeds, including design-infeasible cases.

Each instance has $n=4$ participants. Label proportions follow a Dirichlet distribution with concentration $0.3$, participant sample sizes lie between 800 and 3200, and communication probabilities lie between $0.60$ and $0.95$. A coalition trains a compact convolutional network with approximately 137,000 parameters for $R=20$ rounds, one local epoch per round, minibatches of size 32, and SGD learning rate $0.05$. All coalitions start from the same initialization. Failed uploads follow the Bernoulli model in \eqref{eq:training-probability-space}; when no update arrives, the safe aggregation rule in \eqref{eq:safe-aggregation} leaves the coalition model unchanged. The observed empty-round rates were $20.7\%$ for singletons, $4.1\%$ for pairs, $0.9\%$ for triples, and $0.2\%$ for the grand coalition.

For commensurable economic units, quality is validation accuracy and the benefit map is
\begin{equation}
    B_\lambda(S)
    =\lambda |S|\,
      \E\!\left[
        \left[Q_S(\btheta_S^R)-Q^{\mathrm{ref}}\right]_+
      \right],
    \qquad Q^{\mathrm{ref}}=0.10,
    \label{eq:empirical-benefit-map}
\end{equation}
where $[x]_+=\max\{x,0\}$. Computation and communication costs use the same currency unit and are multiplied by $\gamma$. The benign primary calibration is $(\lambda,\gamma)=(10,1)$; the lean co-primary calibration is $(2,1)$. A sensitivity grid uses $\lambda\in\{2,5,10,20,50\}$ and $\gamma\in\{0.5,1,2\}$. The main pairwise-validation-gain (PVG) mechanism uses five repetitions per coalition. A disjoint 20-repetition PVG estimate supplies the high-replication benchmark. We also evaluate gradient alignment (GA), a transfer-scale arm that multiplies positive-pair caps by $\beta=50$, and the following baselines: grand-coalition FedAvg \cite{McMahanEtAl2017FedAvg}, independent local training, IFCA with two clusters \cite{GhoshEtAl2020IFCA}, and equal-surplus sharing. Under equal-surplus sharing, each member receives $U_i^{\mathrm{eq}}(S)=W(S)/|S|$, so the coalition surplus is divided exactly and the coordinator retains zero. This rule is exactly budget balanced but need not be pairwise representable and therefore has no potential-game guarantee.

For each of seeds 201--205, the experiment evaluates all 15 nonempty coalitions using five mechanism repetitions and 20 disjoint benchmark repetitions, together with five IFCA runs. This gives 1,875 coalition-training runs and 25 IFCA runs, with no failure or skipped run. Exact enumeration over the 15 partitions of four participants certifies welfare and stability relative to each estimated value table. Thus, throughout this section, \emph{certified} means exact for the stated estimated table, not exact for an unknown population value function. Sampling uncertainty is evaluated separately by a preregistered $B=200$ bootstrap that jointly resamples coalition quality and participant accuracy.

The fresh-run provenance chain is fully enforced. Every raw record carries the pre-run commitment; the final raw seal and the versioned publication manifest are externally anchored; and independent verification regenerates the processed results from sealed raw evidence and checks the exact published inventory. The complete chain for namespace \texttt{cifar-essential} passes with zero failures.

\subsection{Primary welfare and stability results}

\Cref{tab:cifar-welfare} reports the benign primary cell on the five-repetition mechanism value table. The mechanism attains the certified welfare optimum for every seed, so the empirical price of stability is exactly one in all five instances. The 95\% confidence interval for mean optimal welfare is $[19.540,21.902]$. Among baselines defined on all five seeds, local training is strongest on average, followed by IFCA evaluated on its own cluster models. Grand-coalition FedAvg is $26.4\%$ below the optimum on average, demonstrating the anti-pooling regime that motivates coalition formation. The independent 20-repetition benchmark-table evaluation is reported separately in the welfare-bound and budget audit below.

\begin{table}[t]
    \centering
    \small
    \caption{Social welfare on the five-repetition mechanism value table at $(\lambda,\gamma)=(10,1)$ across five CIFAR-10 seeds. Values are mean $\pm$ sample standard deviation. The equal-surplus row is conditional on the two seeds for which a Nash-stable partition exists and is therefore not directly comparable with the five-seed rows.}
    \label{tab:cifar-welfare}
    \begin{tabular}{@{}lccc@{}}
        \toprule
        Method & $n$ & Mean $\pm$ sd & 95\% CI \\
        \midrule
        Certified welfare optimum & 5 & $20.721\pm0.951$ & $[19.540,21.902]$ \\
        Pairwise mechanism and better response & 5 & $20.721\pm0.951$ & $[19.540,21.902]$ \\
        Independent local training & 5 & $20.117\pm1.564$ & $[18.174,22.059]$ \\
        IFCA, own cluster models & 5 & $18.233\pm1.505$ & $[16.365,20.102]$ \\
        IFCA partition, common value table & 5 & $17.321\pm2.050$ & $[14.775,19.866]$ \\
        Grand-coalition FedAvg & 5 & $15.259\pm1.963$ & $[12.821,17.696]$ \\
        Equal-surplus best stable, if it exists & 2 & $21.140\pm0.367$ & $[17.844,24.436]$ \\
        \bottomrule
    \end{tabular}
\end{table}

The optimum is all singletons for seeds 202--204. For seeds 201 and 205 it consists of one productive pair and two singletons because one participant has much lower standalone surplus than the others. The designed pair values recover both regimes. In every seed the Nash-stable and individually stable sets each contain a unique partition, that partition equals the welfare optimum, and strict better response reaches it from all nine Nash-dynamics endpoints tested per cell. These comprise seven preregistered initializations and two additional update-order runs from the grand-coalition initialization. Convergence takes at most four moves, with a mean of $1.53$ moves. \Cref{fig:cifar-moves} disaggregates convergence length by initialization and dynamics kind: the grand-coalition initialization is farthest from equilibrium, Nash and destination-consent dynamics behave almost identically, and the greedy-agreement and potential-optimum initializations start at the certified optimum and take zero moves.

\begin{figure}[t]
    \centering
    \includegraphics[width=\linewidth,trim=0 0 0 0.36in,clip]{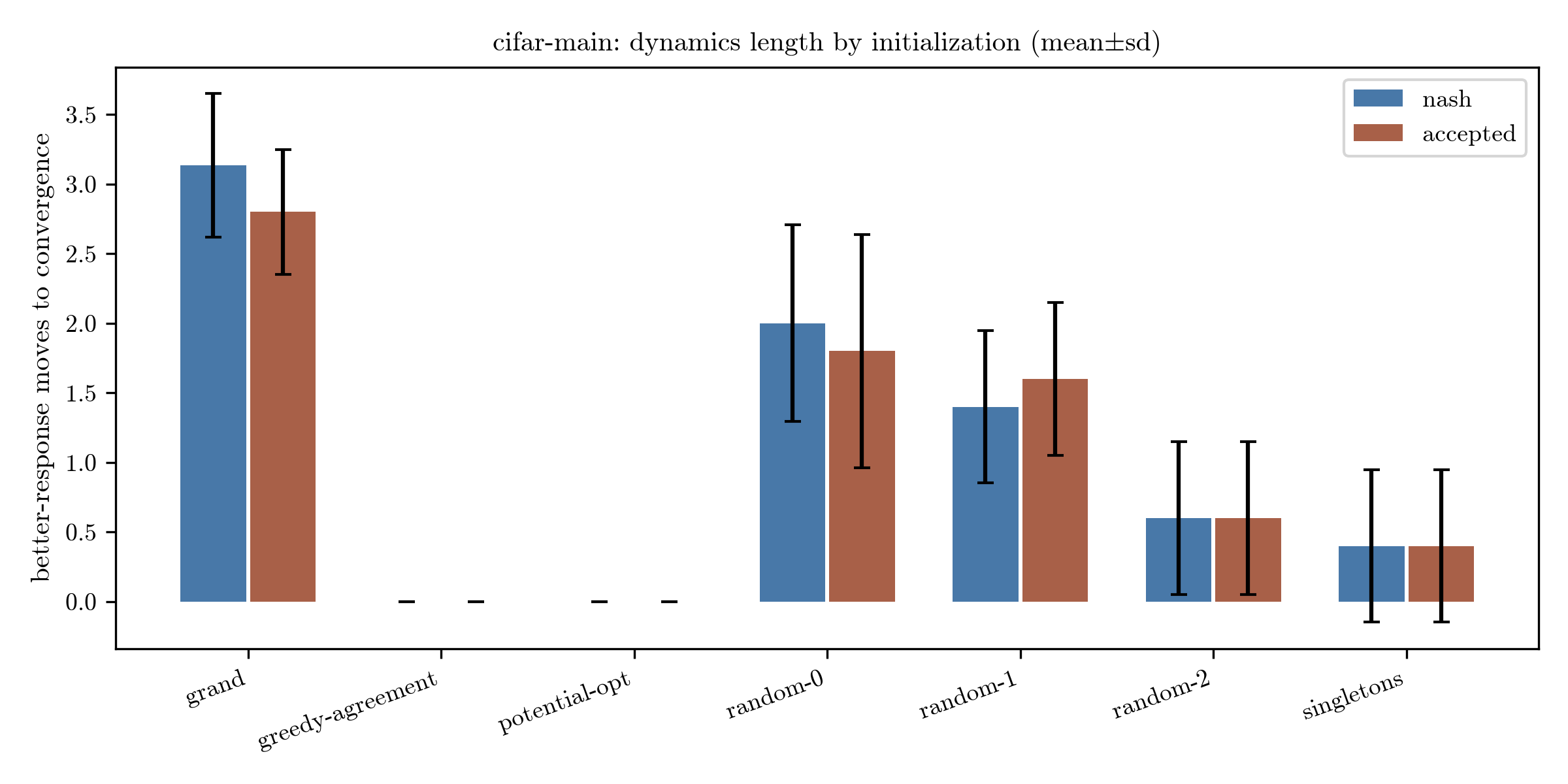}
    \caption{Better-response moves to convergence by initialization and dynamics kind (Nash and destination-consent accepted better response) at the benign cell, mean $\pm$ sample standard deviation over five seeds. The greedy-agreement and potential-optimum initializations take zero moves.}
    \label{fig:cifar-moves}
\end{figure}

This existence result is not shared by the allocation baselines. Equal-surplus sharing has a Nash-stable outcome on only two seeds. On the other three, exhaustive enumeration finds an empty Nash-stable set, and all seven tested strict better-response trajectories cycle. The comparison isolates the practical role of the symmetric pairwise structure: the exact potential guarantees existence and finite convergence, while an apparently natural surplus split does not.

\subsection{Stability costs and retained slack}

The sensitivity grid in \Cref{fig:cifar-grid} separates the benign regime from a lean, cost-sensitive regime. At $(\lambda,\gamma)=(2,1)$, singleton pre-screening certifies design infeasibility for seeds 201 and 205 because participant 2 has negative singleton surplus. The remaining three seeds are feasible but have prices of stability $1.291$, $1.225$, and $1.038$, with relative slack $\delta$ between $1.243$ and $2.317$. Stability can therefore cost up to $22.5\%$ of the certified optimum in the feasible lean instances. This does not contradict the primary result; it shows that the welfare alignment observed at $(10,1)$ is calibration-dependent, as predicted by the worst-case theory.

\begin{figure}[t]
    \centering
    \includegraphics[width=\linewidth,trim=0 0 0 0.38in,clip]{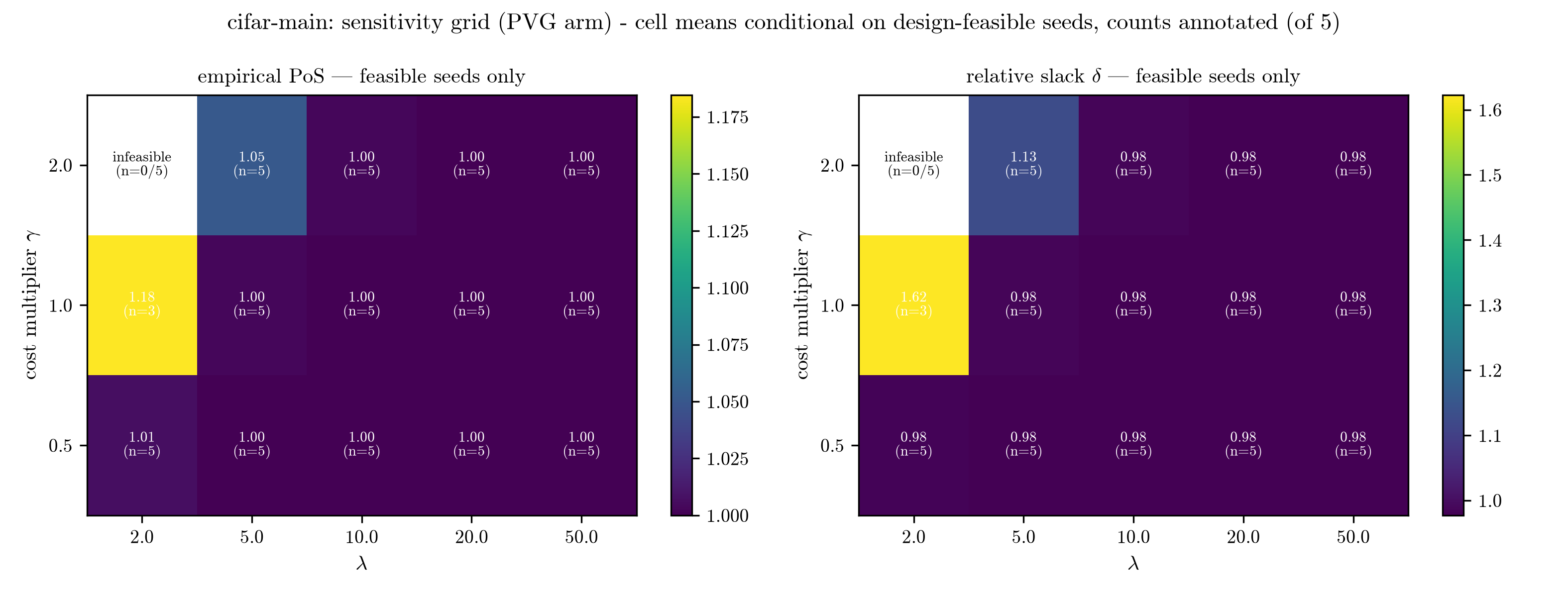}
    \caption{Sensitivity of the empirical price of stability and relative slack for the PVG mechanism. Each entry is the mean over design-feasible seeds, and the annotation gives the number of feasible seeds out of five. Empty cells failed the preregistered singleton feasibility screen.}
    \label{fig:cifar-grid}
\end{figure}

The $\beta=50$ transfer-scale arm probes whether richer transfers can reduce coordinator-retained slack. At the benign cell, it lowers $\delta$ from $0.921$ to $0.517$ for seed 201 and from $0.964$ to $0.538$ for seed 205. The relative-slack guarantee in \Cref{cor:relative-slack} is therefore non-vacuous when the optimum contains a pair. Richer transfers do not automatically improve alignment, however. For seed 201, the transfer-scale arm moves the stable partition away from the optimum and raises the price of stability from $1.000$ to $1.183$; for seed 205, it preserves price of stability one. This is an empirical instance of the distinction between distributing surplus to participants and aligning their potential with social welfare.

\subsection{Pair estimation and sampling uncertainty}

The signs of pair values determine whether a destination coalition accepts an entrant and are therefore more consequential than magnitude error alone. Against the disjoint 20-repetition PVG benchmark, five-repetition PVG obtains the correct sign on 27 of 30 participant pairs, with three false positives and no false negatives. GA obtains 18 of 30 signs, with 12 false positives and no false negatives. Every benchmark pair is sign-confident at the CIFAR-10 effect sizes. PVG's seed-level Pearson correlation with the benchmark ranges from $0.759$ to $0.945$, whereas GA ranges from $0.211$ to $0.909$. \Cref{fig:cifar-estimators} shows that GA is systematically over-optimistic about several benchmark-negative pairs.

\begin{figure}[t]
    \centering
    \includegraphics[width=\linewidth,trim=0 0 0 0.36in,clip]{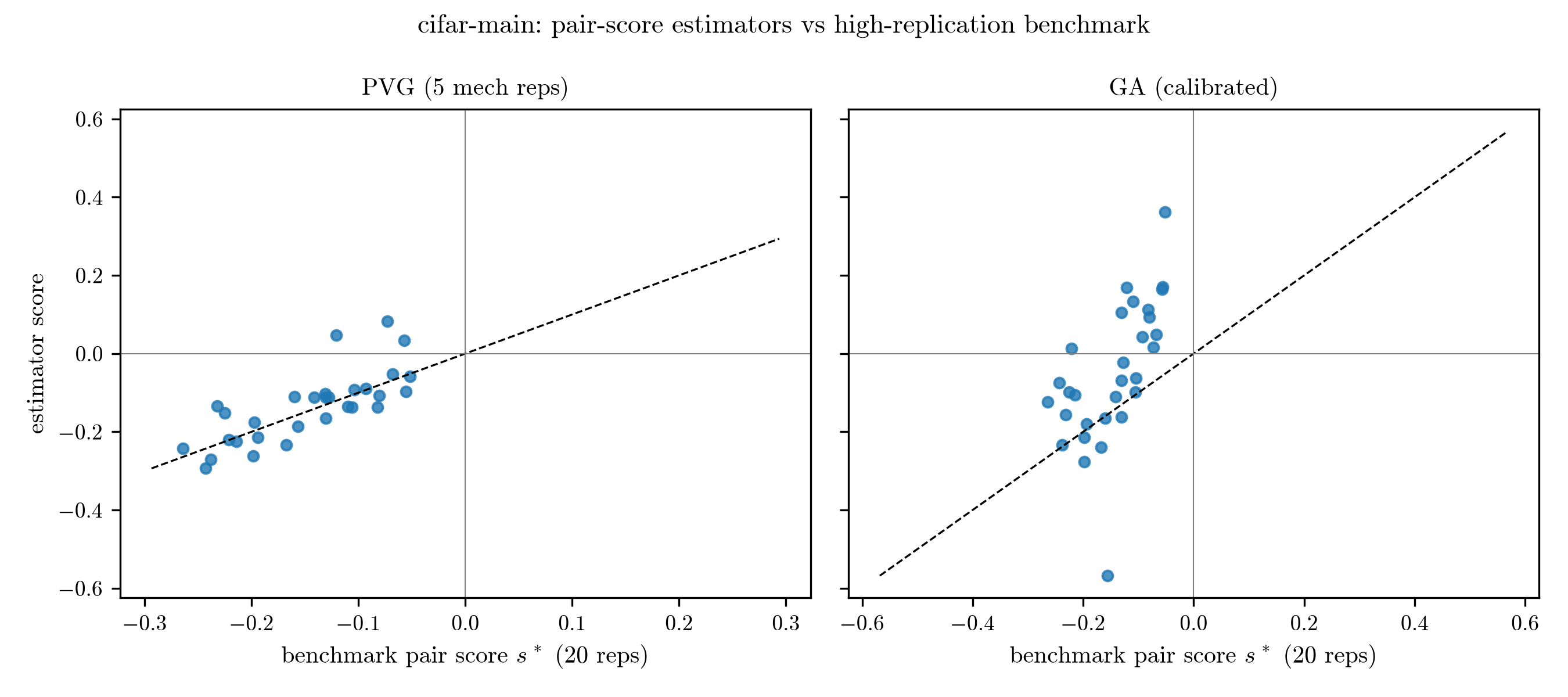}
    \caption{Mechanism-facing pair-score estimators against the disjoint 20-repetition PVG benchmark. Horizontal and vertical zero lines expose sign errors, which directly affect destination acceptance.}
    \label{fig:cifar-estimators}
\end{figure}

The bootstrap qualifies the exact estimated-table claims. At the benign cell, the full-run optimum is reselected in $69\%$ to $100\%$ of feasible resamples, depending on the seed. The 95th percentile of the empirical price of stability is at most $1.0224$, and the best stable welfare exceeds grand-coalition welfare in every feasible resample for every seed. The headline conclusion, stable welfare near the optimum and well above the grand coalition, is therefore robust to the measured value-table uncertainty. At the lean cell, the optimum is reselected in only $29.5\%$ to $55\%$ of resamples. We accordingly interpret the lean-regime partition identities as estimation-sensitive, even though the full-run enumeration is exact for its estimated table. \Cref{fig:cifar-bootstrap-pos} shows the contrast per seed: benign-cell quantile intervals collapse onto one, while lean-cell intervals are wide and the optimum identity is unstable under resampling.

\begin{figure}[t]
    \centering
    \includegraphics[width=\linewidth,trim=0 0 0 0.30in,clip]{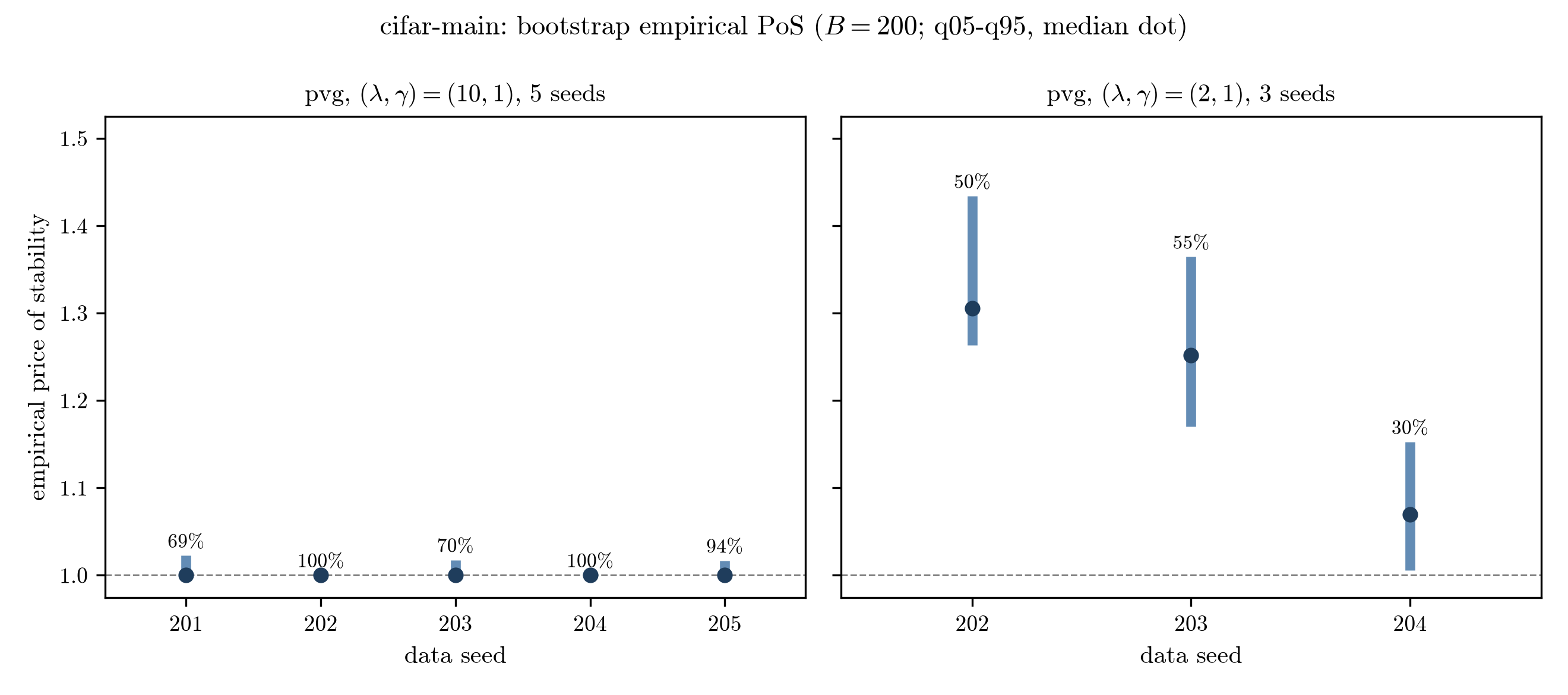}
    \caption{Bootstrap uncertainty of the empirical price of stability at the two co-primary cells ($B=200$ joint resamples of coalition quality and member accuracies; bars span the 5th to 95th percentile over feasible resamples, dots mark medians, labels give the rate at which the full-run optimum is reselected). Seeds 201 and 205 are design-infeasible at the lean cell and carry no lean bootstrap.}
    \label{fig:cifar-bootstrap-pos}
\end{figure}

\subsection{Welfare-bound audit, budget sensitivity, and runtime}

Across all design-feasible cells and all five seeds, the additive slack floor, the relative-slack floor whenever $\delta<1$, and the approximation-plus-stabilization floor in \Cref{thm:approximate-stable-welfare} hold for every audited endpoint on both the mechanism and benchmark value tables. \Cref{fig:cifar-pos-delta} plots the realized price of stability against the realized relative slack for all 92 design-feasible cells: every point with $\delta<1$ lies well below the bound $1/(1-\delta)$ of \Cref{cor:relative-slack}, and the points beyond $\delta=1$, where the bound is vacuous, still realize prices of stability at most $1.291$. The realized agreement ratio spans $[0,1]$, and the greedy agreement initializer finds the exact potential optimum for all five primary instances. This realized ratio is measurable because exhaustive enumeration certifies the potential optimum at $n=4$. Beyond certification scale, an a-priori application of \Cref{thm:approximate-stable-welfare} instead requires an initialization algorithm with a proven agreement-approximation factor. The mean negative-edge mass at the benign cell is $C_v=8.46$, so the negative-edge correction in the end-to-end bound is materially nonzero on these heterogeneous instances.

\begin{figure}[t]
    \centering
    \includegraphics[width=\linewidth,trim=0 0 0 0.325in,clip]{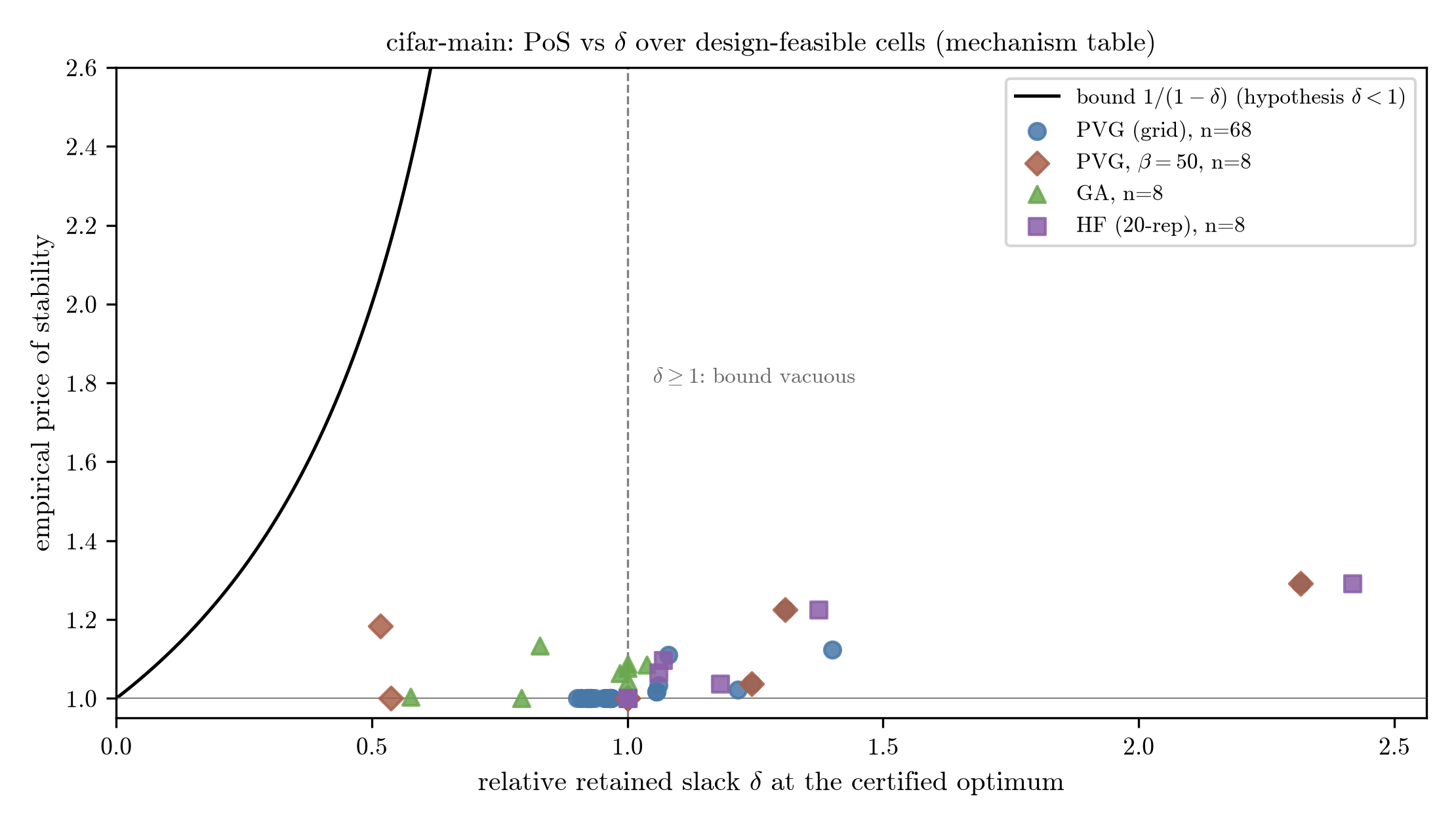}
    \caption{Empirical price of stability against relative retained slack $\delta$ at the certified optimum, for all 92 design-feasible cells on the mechanism value table (all arms and seeds). The curve is the relative-slack bound $1/(1-\delta)$ of \Cref{cor:relative-slack}, valid on its hypothesis region $\delta<1$; points beyond the dashed line at $\delta=1$ carry no relative-slack guarantee. HF denotes the cells in which the 20-repetition benchmark table itself plays the mechanism role.}
    \label{fig:cifar-pos-delta}
\end{figure}

Budget checks on the independent benchmark table identify three estimation-induced violations among 92 design-feasible cells. Two occur in the $\beta=50$ arm for seeds 201 and 205, where the mechanism-table LP drives the productive pair's retained budget to the boundary; benchmark slack is then $-1.04$ and $-1.00$. The third is a smaller $-0.057$ violation for GA at the lean cell of seed 202. The primary $\beta=1$ PVG mechanism has no benchmark-table budget violation. These cases do not invalidate the mechanism-table guarantees, but they show that exact-balance or nearly binding transfer designs are sensitive to coalition-value estimation error.

Coalition training averages 14.6--22.4 seconds per run across seeds on deterministic CPU. The measured sequential-equivalent coalition-training total is 91.3--139.9 minutes per seed, with a mean of 117.6 minutes; operating-system throttling inflated wall-clock time on two seeds without affecting deterministic outputs. The economic stage, including 21 design cells and two 200-resample bootstraps, takes about three seconds per seed. Peak resident memory is approximately 1.1 GiB per worker. All 92 design-feasible processed cells satisfy the algebraic identities to a maximum relative residual of $1.8\times10^{-15}$ (the 13 screened cells carry no enumeration to audit), and all 1{,}472 tested dynamics endpoints pass their stated stability checks.

\subsection{Scope and limitations}

The four-participant scale is a deliberate tradeoff in favor of exact estimated-table certification. Evaluating coalition values requires training all $2^n-1$ nonempty coalitions, while exact welfare and stability certification requires searching a number of partitions that grows according to the Bell numbers. At $n=4$, both inventories contain 15 objects, and every coalition can be evaluated under 25 independent repetitions within the preregistered deterministic-CPU compute envelope. The resulting claims are exact for each estimated table, but they do not establish behavior at larger population sizes. The study is also limited to CIFAR-10 image classification, one primary non-IID concentration of $0.3$, five analysis seeds, and IFCA as the only learned-clustering baseline. Wider claims require additional tasks and modalities, more heterogeneity regimes and baselines, and scalable value-estimation and optimization methods that do not enumerate every coalition and partition.

\section{Scope of the Mechanism-Design Claim}
The current model assumes that the coordinator knows or estimates $B(S)$, $C_0(S)$, and $d_i(S)$. It designs transfers but does not yet elicit private types. Accordingly, the accurate description is \emph{incentive allocation and coalition formation}, not a truthful direct mechanism. If costs, reliability, data quality, or privacy preferences are privately reported, a stronger journal extension must define the message space and prove incentive compatibility in addition to individual rationality and budget feasibility.

\section{Conclusion}
This journal framework makes coalition membership consequential for federated learning, separates economic primitives from transfers, and provides a correct stability foundation. Symmetric pairwise allocations yield an exact potential game, ensuring existence and finite better-response convergence for Nash stability and destination-consent-based individual stability. Bounded pair incentives admit an exact feasibility characterization, and submodular retained slack permits polynomial oracle-time verification in an exact tractability regime. Coordinator-retained budget slack at the welfare optimum controls the additive gap to a stable potential optimum, while relative slack gives a multiplicative price-of-stability bound that is asymptotically tight. Exact balance gives welfare-optimal stability only on the exactly pairwise-representable surplus class; without alignment, welfare loss is unbounded. Global potential optimization is weighted maximum-agreement correlation clustering. Approximation followed by stability post-processing preserves agreement value and satisfies an end-to-end welfare bound whose losses are governed by retained slack and negative-edge mass; the negative-edge correction is attained in a four-participant construction. The CIFAR-10 study shows that the pairwise mechanism can attain the certified estimated-table optimum with rapid decentralized convergence in a benign regime, while lean calibrations reveal genuine stability costs and nearly balanced transfers reveal sensitivity to value-estimation error. It also confirms that pair-sign estimation, retained slack, and negative-edge mass are operational diagnostics rather than proof artifacts. Remaining work concerns larger coalition populations, richer learning tasks, private-type elicitation, FL-specific alignment guarantees, and convergence complexity.

\backmatter

\section*{Declarations}

\noindent\textbf{Funding.}\quad
No funding was received for conducting this study.

\noindent\textbf{Competing interests.}\quad
The author has no competing interests to declare that are relevant to the
content of this article.

\noindent\textbf{Ethics approval.}\quad Not applicable. This study involves no
human participants, human data, or animals; all experiments use public
benchmark datasets (CIFAR-10).

\noindent\textbf{Consent to participate.}\quad Not applicable.

\noindent\textbf{Consent for publication.}\quad Not applicable.

\noindent\textbf{Data availability.}\quad
This study uses the publicly available CIFAR-10 dataset. The processed results
and the sealed raw experimental records are included in the reproducibility
package, which is openly available in a public GitHub repository,
\url{https://github.com/cengishasan/clustered-fl-coalition-formation}, and archived under a persistent DOI,
\url{https://doi.org/10.5281/zenodo.21428635}.

\noindent\textbf{Code availability.}\quad
The source code, configurations, preregistration contract, and a verification
script that reproduces all reported results from the sealed raw records are
available in the same GitHub repository and archived under the same DOI.

\noindent\textbf{Author contributions.}\quad
C.H. is the sole author and conducted all aspects of the work, including the
theory, experiments, and writing.

\bibliography{references}


\begin{thebibliography}{35}
\ifx \bisbn   \undefined \def \bisbn  #1{ISBN #1}\fi
\ifx \binits  \undefined \def \binits#1{#1}\fi
\ifx \bauthor  \undefined \def \bauthor#1{#1}\fi
\ifx \batitle  \undefined \def \batitle#1{#1}\fi
\ifx \bjtitle  \undefined \def \bjtitle#1{#1}\fi
\ifx \bvolume  \undefined \def \bvolume#1{\textbf{#1}}\fi
\ifx \byear  \undefined \def \byear#1{#1}\fi
\ifx \bissue  \undefined \def \bissue#1{#1}\fi
\ifx \bfpage  \undefined \def \bfpage#1{#1}\fi
\ifx \blpage  \undefined \def \blpage #1{#1}\fi
\ifx \burl  \undefined \def \burl#1{\textsf{#1}}\fi
\ifx \doiurl  \undefined \def \doiurl#1{\url{https://doi.org/#1}}\fi
\ifx \betal  \undefined \def \betal{\textit{et al.}}\fi
\ifx \binstitute  \undefined \def \binstitute#1{#1}\fi
\ifx \binstitutionaled  \undefined \def \binstitutionaled#1{#1}\fi
\ifx \bctitle  \undefined \def \bctitle#1{#1}\fi
\ifx \beditor  \undefined \def \beditor#1{#1}\fi
\ifx \bpublisher  \undefined \def \bpublisher#1{#1}\fi
\ifx \bbtitle  \undefined \def \bbtitle#1{#1}\fi
\ifx \bedition  \undefined \def \bedition#1{#1}\fi
\ifx \bseriesno  \undefined \def \bseriesno#1{#1}\fi
\ifx \blocation  \undefined \def \blocation#1{#1}\fi
\ifx \bsertitle  \undefined \def \bsertitle#1{#1}\fi
\ifx \bsnm \undefined \def \bsnm#1{#1}\fi
\ifx \bsuffix \undefined \def \bsuffix#1{#1}\fi
\ifx \bparticle \undefined \def \bparticle#1{#1}\fi
\ifx \barticle \undefined \def \barticle#1{#1}\fi
\bibcommenthead
\ifx \bconfdate \undefined \def \bconfdate #1{#1}\fi
\ifx \botherref \undefined \def \botherref #1{#1}\fi
\ifx \url \undefined \def \url#1{\textsf{#1}}\fi
\ifx \bchapter \undefined \def \bchapter#1{#1}\fi
\ifx \bbook \undefined \def \bbook#1{#1}\fi
\ifx \bcomment \undefined \def \bcomment#1{#1}\fi
\ifx \oauthor \undefined \def \oauthor#1{#1}\fi
\ifx \citeauthoryear \undefined \def \citeauthoryear#1{#1}\fi
\ifx \endbibitem  \undefined \def \endbibitem {}\fi
\ifx \bconflocation  \undefined \def \bconflocation#1{#1}\fi
\ifx \arxivurl  \undefined \def \arxivurl#1{\textsf{#1}}\fi
\csname PreBibitemsHook\endcsname

\bibitem[\protect\citeauthoryear{McMahan et~al.}{2017}]{McMahanEtAl2017FedAvg}
\begin{bchapter}
\bauthor{\bsnm{McMahan}, \binits{H.B.}},
\bauthor{\bsnm{Moore}, \binits{E.}},
\bauthor{\bsnm{Ramage}, \binits{D.}},
\bauthor{\bsnm{Hampson}, \binits{S.}},
\bauthor{\bsnm{Arcas}, \binits{B.A.}}:
\bctitle{Communication-efficient learning of deep networks from decentralized
  data}.
In: \bbtitle{Proceedings of the 20th International Conference on Artificial
  Intelligence and Statistics},
pp. \bfpage{1273}--\blpage{1282}
(\byear{2017})
\end{bchapter}
\endbibitem

\bibitem[\protect\citeauthoryear{Li et~al.}{2020}]{LiEtAl2020Survey}
\begin{barticle}
\bauthor{\bsnm{Li}, \binits{T.}},
\bauthor{\bsnm{Sahu}, \binits{A.K.}},
\bauthor{\bsnm{Talwalkar}, \binits{A.}},
\bauthor{\bsnm{Smith}, \binits{V.}}:
\batitle{Federated learning: Challenges, methods, and future directions}.
\bjtitle{IEEE Signal Processing Magazine}
\bvolume{37}(\bissue{3}),
\bfpage{50}--\blpage{60}
(\byear{2020})
\doiurl{10.1109/MSP.2020.2975749}
\end{barticle}
\endbibitem

\bibitem[\protect\citeauthoryear{Kairouz
  et~al.}{2021}]{KairouzEtAl2021OpenProblems}
\begin{barticle}
\bauthor{\bsnm{Kairouz}, \binits{P.}},
\bauthor{\bsnm{McMahan}, \binits{H.B.}}, \betal:
\batitle{Advances and open problems in federated learning}.
\bjtitle{Foundations and Trends in Machine Learning}
\bvolume{14}(\bissue{1--2}),
\bfpage{1}--\blpage{210}
(\byear{2021})
\doiurl{10.1561/2200000083}
\end{barticle}
\endbibitem

\bibitem[\protect\citeauthoryear{Bogomolnaia and
  Jackson}{2002}]{BogomolnaiaJackson2002}
\begin{barticle}
\bauthor{\bsnm{Bogomolnaia}, \binits{A.}},
\bauthor{\bsnm{Jackson}, \binits{M.O.}}:
\batitle{The stability of hedonic coalition structures}.
\bjtitle{Games and Economic Behavior}
\bvolume{38}(\bissue{2}),
\bfpage{201}--\blpage{230}
(\byear{2002})
\doiurl{10.1006/game.2001.0877}
\end{barticle}
\endbibitem

\bibitem[\protect\citeauthoryear{Hajdukov{\'a}}{2006}]{Hajdukova2006}
\begin{barticle}
\bauthor{\bsnm{Hajdukov{\'a}}, \binits{J.}}:
\batitle{Coalition formation games: A survey}.
\bjtitle{International Game Theory Review}
\bvolume{8}(\bissue{4}),
\bfpage{613}--\blpage{641}
(\byear{2006})
\doiurl{10.1142/S0219198906001144}
\end{barticle}
\endbibitem

\bibitem[\protect\citeauthoryear{Aziz and Brandl}{2012}]{AzizBrandl2012}
\begin{bchapter}
\bauthor{\bsnm{Aziz}, \binits{H.}},
\bauthor{\bsnm{Brandl}, \binits{F.}}:
\bctitle{Existence of stability in hedonic coalition formation games}.
In: \bbtitle{Proceedings of the 11th International Conference on Autonomous
  Agents and Multiagent Systems}
(\byear{2012})
\end{bchapter}
\endbibitem

\bibitem[\protect\citeauthoryear{Hasan}{2021}]{Hasan2021HedonicFL}
\begin{botherref}
\oauthor{\bsnm{Hasan}, \binits{C.}}:
Incentive Mechanism Design for Federated Learning: Hedonic Game Approach.
arXiv:2101.09673
(2021)
\end{botherref}
\endbibitem

\bibitem[\protect\citeauthoryear{Le et~al.}{2021}]{LeEtAl2021Auction}
\begin{barticle}
\bauthor{\bsnm{Le}, \binits{T.H.T.}},
\bauthor{\bsnm{Tran}, \binits{N.H.}},
\bauthor{\bsnm{Tun}, \binits{Y.K.}},
\bauthor{\bsnm{Nguyen}, \binits{M.N.H.}},
\bauthor{\bsnm{Pandey}, \binits{S.R.}},
\bauthor{\bsnm{Han}, \binits{Z.}},
\bauthor{\bsnm{Hong}, \binits{C.S.}}:
\batitle{An incentive mechanism for federated learning in wireless cellular
  networks: An auction approach}.
\bjtitle{IEEE Transactions on Wireless Communications}
\bvolume{20}(\bissue{8}),
\bfpage{4874}--\blpage{4887}
(\byear{2021})
\doiurl{10.1109/TWC.2021.3062708}
\end{barticle}
\endbibitem

\bibitem[\protect\citeauthoryear{Ding et~al.}{2020}]{DingEtAl2020PrivateInfo}
\begin{bchapter}
\bauthor{\bsnm{Ding}, \binits{N.}},
\bauthor{\bsnm{Fang}, \binits{Z.}},
\bauthor{\bsnm{Huang}, \binits{J.}}:
\bctitle{Incentive mechanism design for federated learning with
  multi-dimensional private information}.
In: \bbtitle{2020 18th International Symposium on Modeling and Optimization in
  Mobile, Ad Hoc, and Wireless Networks},
pp. \bfpage{1}--\blpage{8}
(\byear{2020})
\end{bchapter}
\endbibitem

\bibitem[\protect\citeauthoryear{Sarikaya and
  Ercetin}{2020}]{SarikayaErcetin2020}
\begin{barticle}
\bauthor{\bsnm{Sarikaya}, \binits{Y.}},
\bauthor{\bsnm{Ercetin}, \binits{O.}}:
\batitle{Motivating workers in federated learning: A stackelberg game
  perspective}.
\bjtitle{IEEE Networking Letters}
\bvolume{2}(\bissue{1}),
\bfpage{23}--\blpage{27}
(\byear{2020})
\doiurl{10.1109/LNET.2019.2947144}
\end{barticle}
\endbibitem

\bibitem[\protect\citeauthoryear{Zhan and Zhang}{2020}]{ZhanZhang2020DRL}
\begin{bchapter}
\bauthor{\bsnm{Zhan}, \binits{Y.}},
\bauthor{\bsnm{Zhang}, \binits{J.}}:
\bctitle{An incentive mechanism design for efficient edge learning by deep
  reinforcement learning approach}.
In: \bbtitle{IEEE INFOCOM 2020 - IEEE Conference on Computer Communications},
pp. \bfpage{2489}--\blpage{2498}
(\byear{2020}).
\doiurl{10.1109/INFOCOM41043.2020.9155268}
\end{bchapter}
\endbibitem

\bibitem[\protect\citeauthoryear{Kim}{2020}]{Kim2020IncentiveDP}
\begin{barticle}
\bauthor{\bsnm{Kim}, \binits{S.}}:
\batitle{Incentive design and differential privacy based federated learning: A
  mechanism design perspective}.
\bjtitle{IEEE Access}
\bvolume{8},
\bfpage{187317}--\blpage{187325}
(\byear{2020})
\doiurl{10.1109/ACCESS.2020.3030888}
\end{barticle}
\endbibitem

\bibitem[\protect\citeauthoryear{Tu et~al.}{2022}]{TuEtAl2022Incentive}
\begin{barticle}
\bauthor{\bsnm{Tu}, \binits{X.}},
\bauthor{\bsnm{Zhu}, \binits{K.}},
\bauthor{\bsnm{Luong}, \binits{N.C.}},
\bauthor{\bsnm{Niyato}, \binits{D.}},
\bauthor{\bsnm{Zhang}, \binits{Y.}},
\bauthor{\bsnm{Li}, \binits{J.}}:
\batitle{Incentive mechanisms for federated learning: From economic and game
  theoretic perspective}.
\bjtitle{IEEE Transactions on Cognitive Communications and Networking}
\bvolume{8}(\bissue{3}),
\bfpage{1566}--\blpage{1593}
(\byear{2022})
\doiurl{10.1109/TCCN.2022.3177522}
\end{barticle}
\endbibitem

\bibitem[\protect\citeauthoryear{Wu et~al.}{2025}]{WuEtAl2025Diversity}
\begin{barticle}
\bauthor{\bsnm{Wu}, \binits{X.}},
\bauthor{\bsnm{Lin}, \binits{Y.}},
\bauthor{\bsnm{Zhong}, \binits{H.}},
\bauthor{\bsnm{Tao}, \binits{J.}},
\bauthor{\bsnm{Gu}, \binits{Y.}},
\bauthor{\bsnm{Shen}, \binits{S.}},
\bauthor{\bsnm{Yu}, \binits{S.}}:
\batitle{A diversity-aware incentive mechanism for cross-silo federated
  learning with budget constraint}.
\bjtitle{Knowledge-Based Systems}
\bvolume{315},
\bfpage{113212}
(\byear{2025})
\doiurl{10.1016/j.knosys.2025.113212}
\end{barticle}
\endbibitem

\bibitem[\protect\citeauthoryear{Donahue and
  Kleinberg}{2021}]{DonahueKleinberg2021ModelSharing}
\begin{barticle}
\bauthor{\bsnm{Donahue}, \binits{K.}},
\bauthor{\bsnm{Kleinberg}, \binits{J.}}:
\batitle{Model-sharing games: Analyzing federated learning under voluntary
  participation}.
\bjtitle{Proceedings of the AAAI Conference on Artificial Intelligence}
\bvolume{35}(\bissue{6}),
\bfpage{5303}--\blpage{5311}
(\byear{2021})
\doiurl{10.1609/aaai.v35i6.16669}
\end{barticle}
\endbibitem

\bibitem[\protect\citeauthoryear{Blum et~al.}{2021}]{BlumEtAl2021Collaboration}
\begin{bchapter}
\bauthor{\bsnm{Blum}, \binits{A.}},
\bauthor{\bsnm{Haghtalab}, \binits{N.}},
\bauthor{\bsnm{Phillips}, \binits{R.L.}},
\bauthor{\bsnm{Shao}, \binits{H.}}:
\bctitle{One for one, or all for all: Equilibria and optimality of
  collaboration in federated learning}.
In: \bbtitle{Proceedings of the 38th International Conference on Machine
  Learning},
pp. \bfpage{1005}--\blpage{1014}
(\byear{2021})
\end{bchapter}
\endbibitem

\bibitem[\protect\citeauthoryear{Chaudhury
  et~al.}{2022}]{ChaudhuryEtAl2022Core}
\begin{bchapter}
\bauthor{\bsnm{Chaudhury}, \binits{B.R.}},
\bauthor{\bsnm{Li}, \binits{L.}},
\bauthor{\bsnm{Kang}, \binits{M.}},
\bauthor{\bsnm{Li}, \binits{B.}},
\bauthor{\bsnm{Mehta}, \binits{R.}}:
\bctitle{Fairness in federated learning via core-stability}.
In: \bbtitle{Advances in Neural Information Processing Systems},
vol. \bseriesno{35}
(\byear{2022}).
\burl{https://proceedings.neurips.cc/paper_files/paper/2022/hash/25e92e33ac8c35fd49f394c37f21b6da-Abstract-Conference.html}
\end{bchapter}
\endbibitem

\bibitem[\protect\citeauthoryear{Chen et~al.}{2024}]{ChenEtAl2024FedEgoists}
\begin{bchapter}
\bauthor{\bsnm{Chen}, \binits{M.}},
\bauthor{\bsnm{Wu}, \binits{X.}},
\bauthor{\bsnm{Tang}, \binits{X.}},
\bauthor{\bsnm{He}, \binits{T.}},
\bauthor{\bsnm{Ong}, \binits{Y.-S.}},
\bauthor{\bsnm{Liu}, \binits{Q.}},
\bauthor{\bsnm{Lao}, \binits{Q.}},
\bauthor{\bsnm{Yu}, \binits{H.}}:
\bctitle{Free-rider and conflict aware collaboration formation for cross-silo
  federated learning}.
In: \bbtitle{Advances in Neural Information Processing Systems},
vol. \bseriesno{37}
(\byear{2024}).
\burl{https://proceedings.neurips.cc/paper_files/paper/2024/hash/62ffefbe8dfc8548d22564f3c1d21488-Abstract-Conference.html}
\end{bchapter}
\endbibitem

\bibitem[\protect\citeauthoryear{Singer}{2010}]{Singer2010BudgetFeasible}
\begin{bchapter}
\bauthor{\bsnm{Singer}, \binits{Y.}}:
\bctitle{Budget feasible mechanisms}.
In: \bbtitle{2010 IEEE 51st Annual Symposium on Foundations of Computer
  Science},
pp. \bfpage{765}--\blpage{774}
(\byear{2010})
\end{bchapter}
\endbibitem

\bibitem[\protect\citeauthoryear{Ghosh et~al.}{2020}]{GhoshEtAl2020IFCA}
\begin{bchapter}
\bauthor{\bsnm{Ghosh}, \binits{A.}},
\bauthor{\bsnm{Chung}, \binits{J.}},
\bauthor{\bsnm{Yin}, \binits{D.}},
\bauthor{\bsnm{Ramchandran}, \binits{K.}}:
\bctitle{An efficient framework for clustered federated learning}.
In: \bbtitle{Advances in Neural Information Processing Systems},
vol. \bseriesno{33}
(\byear{2020}).
\burl{https://proceedings.neurips.cc/paper/2020/hash/e32cc80bf07915058ce90722ee17bb71-Abstract.html}
\end{bchapter}
\endbibitem

\bibitem[\protect\citeauthoryear{Sattler et~al.}{2021}]{SattlerEtAl2021CFL}
\begin{barticle}
\bauthor{\bsnm{Sattler}, \binits{F.}},
\bauthor{\bsnm{M{\"u}ller}, \binits{K.-R.}},
\bauthor{\bsnm{Samek}, \binits{W.}}:
\batitle{Clustered federated learning: Model-agnostic distributed multitask
  optimization under privacy constraints}.
\bjtitle{IEEE Transactions on Neural Networks and Learning Systems}
\bvolume{32}(\bissue{8}),
\bfpage{3710}--\blpage{3722}
(\byear{2021})
\doiurl{10.1109/TNNLS.2020.3015958}
\end{barticle}
\endbibitem

\bibitem[\protect\citeauthoryear{Aziz and
  Savani}{2016}]{AzizSavani2016HedonicChapter}
\begin{bchapter}
\bauthor{\bsnm{Aziz}, \binits{H.}},
\bauthor{\bsnm{Savani}, \binits{R.}}:
\bctitle{Hedonic games}.
In: \beditor{\bsnm{Brandt}, \binits{F.}},
\beditor{\bsnm{Conitzer}, \binits{V.}},
\beditor{\bsnm{Endriss}, \binits{U.}},
\beditor{\bsnm{Lang}, \binits{J.}},
\beditor{\bsnm{Procaccia}, \binits{A.D.}} (eds.)
\bbtitle{Handbook of Computational Social Choice}.
\bpublisher{Cambridge University Press}, \blocation{???}
(\byear{2016}).
\bcomment{Chap. 15}
\end{bchapter}
\endbibitem

\bibitem[\protect\citeauthoryear{Brandt
  et~al.}{2023}]{BrandtBullingerWilczynski2023}
\begin{barticle}
\bauthor{\bsnm{Brandt}, \binits{F.}},
\bauthor{\bsnm{Bullinger}, \binits{M.}},
\bauthor{\bsnm{Wilczynski}, \binits{A.}}:
\batitle{Reaching individually stable coalition structures}.
\bjtitle{ACM Transactions on Economics and Computation}
\bvolume{11}(\bissue{1--2}),
\bfpage{4}--\blpage{1465}
(\byear{2023})
\doiurl{10.1145/3588753}
\end{barticle}
\endbibitem

\bibitem[\protect\citeauthoryear{Bil{\`o} et~al.}{2018}]{BiloEtAl2018FHG}
\begin{barticle}
\bauthor{\bsnm{Bil{\`o}}, \binits{V.}},
\bauthor{\bsnm{Fanelli}, \binits{A.}},
\bauthor{\bsnm{Flammini}, \binits{M.}},
\bauthor{\bsnm{Monaco}, \binits{G.}},
\bauthor{\bsnm{Moscardelli}, \binits{L.}}:
\batitle{Nash stable outcomes in fractional hedonic games: Existence,
  efficiency and computation}.
\bjtitle{Journal of Artificial Intelligence Research}
\bvolume{62},
\bfpage{315}--\blpage{371}
(\byear{2018})
\doiurl{10.1613/jair.1.11211}
\end{barticle}
\endbibitem

\bibitem[\protect\citeauthoryear{Donahue and
  Kleinberg}{2021}]{DonahueKleinberg2021Optimality}
\begin{bchapter}
\bauthor{\bsnm{Donahue}, \binits{K.}},
\bauthor{\bsnm{Kleinberg}, \binits{J.M.}}:
\bctitle{Optimality and stability in federated learning: A game-theoretic
  approach}.
In: \bbtitle{Advances in Neural Information Processing Systems},
vol. \bseriesno{34}
(\byear{2021}).
\burl{https://proceedings.neurips.cc/paper/2021/hash/09a5e2a11bea20817477e0b1dfe2cc21-Abstract.html}
\end{bchapter}
\endbibitem

\bibitem[\protect\citeauthoryear{Ng et~al.}{2022}]{NgEtAl2022ReputationHedonic}
\begin{barticle}
\bauthor{\bsnm{Ng}, \binits{J.S.}},
\bauthor{\bsnm{Lim}, \binits{W.Y.B.}},
\bauthor{\bsnm{Xiong}, \binits{Z.}},
\bauthor{\bsnm{Cao}, \binits{X.}},
\bauthor{\bsnm{Jin}, \binits{J.}},
\bauthor{\bsnm{Niyato}, \binits{D.}},
\bauthor{\bsnm{Leung}, \binits{C.}},
\bauthor{\bsnm{Miao}, \binits{C.}}:
\batitle{Reputation-aware hedonic coalition formation for efficient serverless
  hierarchical federated learning}.
\bjtitle{IEEE Transactions on Parallel and Distributed Systems}
\bvolume{33}(\bissue{11}),
\bfpage{2675}--\blpage{2686}
(\byear{2022})
\end{barticle}
\endbibitem

\bibitem[\protect\citeauthoryear{Chen et~al.}{2025}]{ChenEtAl2025DualGFL}
\begin{barticle}
\bauthor{\bsnm{Chen}, \binits{X.}},
\bauthor{\bsnm{Zhou}, \binits{X.}},
\bauthor{\bsnm{Zhang}, \binits{S.}},
\bauthor{\bsnm{Sun}, \binits{M.}}:
\batitle{{DualGFL}: Federated learning with a dual-level coalition-auction
  game}.
\bjtitle{Proceedings of the AAAI Conference on Artificial Intelligence}
\bvolume{39}(\bissue{15}),
\bfpage{15904}--\blpage{15912}
(\byear{2025})
\doiurl{10.1609/aaai.v39i15.33746}
\end{barticle}
\endbibitem

\bibitem[\protect\citeauthoryear{Liu et~al.}{2025}]{LiuEtAl2025AAV}
\begin{barticle}
\bauthor{\bsnm{Liu}, \binits{J.}},
\bauthor{\bsnm{Li}, \binits{X.}},
\bauthor{\bsnm{Xu}, \binits{Y.}},
\bauthor{\bsnm{Lyu}, \binits{C.}},
\bauthor{\bsnm{Wang}, \binits{Y.}},
\bauthor{\bsnm{Liu}, \binits{X.}}:
\batitle{Hedonic coalition formation game and contract-based federated learning
  in {AAV}-assisted internet of things}.
\bjtitle{IEEE Internet of Things Journal}
\bvolume{12}(\bissue{9}),
\bfpage{11258}--\blpage{11272}
(\byear{2025})
\doiurl{10.1109/JIOT.2025.3543917}
\end{barticle}
\endbibitem

\bibitem[\protect\citeauthoryear{Bansal et~al.}{2004}]{BansalBlumChawla2004}
\begin{barticle}
\bauthor{\bsnm{Bansal}, \binits{N.}},
\bauthor{\bsnm{Blum}, \binits{A.}},
\bauthor{\bsnm{Chawla}, \binits{S.}}:
\batitle{Correlation clustering}.
\bjtitle{Machine Learning}
\bvolume{56}(\bissue{1--3}),
\bfpage{89}--\blpage{113}
(\byear{2004})
\doiurl{10.1023/B:MACH.0000033116.57574.95}
\end{barticle}
\endbibitem

\bibitem[\protect\citeauthoryear{Swamy}{2004}]{Swamy2004}
\begin{bchapter}
\bauthor{\bsnm{Swamy}, \binits{C.}}:
\bctitle{Correlation clustering: Maximizing agreements via semidefinite
  programming}.
In: \bbtitle{Proceedings of the Fifteenth Annual ACM-SIAM Symposium on Discrete
  Algorithms},
pp. \bfpage{526}--\blpage{527}
(\byear{2004})
\end{bchapter}
\endbibitem

\bibitem[\protect\citeauthoryear{Dr{\`e}ze and
  Greenberg}{1980}]{DrezeGreenberg1980}
\begin{barticle}
\bauthor{\bsnm{Dr{\`e}ze}, \binits{J.H.}},
\bauthor{\bsnm{Greenberg}, \binits{J.}}:
\batitle{Hedonic coalitions: Optimality and stability}.
\bjtitle{Econometrica}
\bvolume{48}(\bissue{4}),
\bfpage{987}--\blpage{1003}
(\byear{1980})
\doiurl{10.2307/1912943}
\end{barticle}
\endbibitem

\bibitem[\protect\citeauthoryear{Monderer and
  Shapley}{1996}]{MondererShapley1996}
\begin{barticle}
\bauthor{\bsnm{Monderer}, \binits{D.}},
\bauthor{\bsnm{Shapley}, \binits{L.S.}}:
\batitle{Potential games}.
\bjtitle{Games and Economic Behavior}
\bvolume{14}(\bissue{1}),
\bfpage{124}--\blpage{143}
(\byear{1996})
\doiurl{10.1006/game.1996.0044}
\end{barticle}
\endbibitem

\bibitem[\protect\citeauthoryear{Anshelevich
  et~al.}{2008}]{AnshelevichEtAl2008PoS}
\begin{barticle}
\bauthor{\bsnm{Anshelevich}, \binits{E.}},
\bauthor{\bsnm{Dasgupta}, \binits{A.}},
\bauthor{\bsnm{Kleinberg}, \binits{J.}},
\bauthor{\bsnm{Tardos}, \binits{{\'E}.}},
\bauthor{\bsnm{Wexler}, \binits{T.}},
\bauthor{\bsnm{Roughgarden}, \binits{T.}}:
\batitle{The price of stability for network design with fair cost allocation}.
\bjtitle{SIAM Journal on Computing}
\bvolume{38}(\bissue{4}),
\bfpage{1602}--\blpage{1623}
(\byear{2008})
\doiurl{10.1137/070680096}
\end{barticle}
\endbibitem

\bibitem[\protect\citeauthoryear{Schrijver}{2000}]{Schrijver2000Submodular}
\begin{barticle}
\bauthor{\bsnm{Schrijver}, \binits{A.}}:
\batitle{A combinatorial algorithm minimizing submodular functions in strongly
  polynomial time}.
\bjtitle{Journal of Combinatorial Theory, Series B}
\bvolume{80}(\bissue{2}),
\bfpage{346}--\blpage{355}
(\byear{2000})
\doiurl{10.1006/jctb.2000.1989}
\end{barticle}
\endbibitem

\bibitem[\protect\citeauthoryear{Aziz et~al.}{2013}]{AzizBrandtSeedig2013}
\begin{barticle}
\bauthor{\bsnm{Aziz}, \binits{H.}},
\bauthor{\bsnm{Brandt}, \binits{F.}},
\bauthor{\bsnm{Seedig}, \binits{H.G.}}:
\batitle{Computing desirable partitions in additively separable hedonic games}.
\bjtitle{Artificial Intelligence}
\bvolume{195},
\bfpage{316}--\blpage{334}
(\byear{2013})
\doiurl{10.1016/j.artint.2012.09.006}
\end{barticle}
\endbibitem

\end{thebibliography}

\end{document}